\documentclass[conference,letterpaper]{IEEEtran}

\IEEEoverridecommandlockouts

\usepackage{amsmath,amssymb,amsthm}
\usepackage{dsfont}
\usepackage{cite}
\usepackage{url}
\usepackage{tikz}
\usepackage[hidelinks]{hyperref}
\usepackage{psfrag}
\usepackage{algorithm}
\usepackage{algorithmic}
\usepackage[normalem]{ulem}

\PassOptionsToPackage{hyphens}{url}
\PassOptionsToPackage{usenames,dvipsnames}{xcolor}

\definecolor{TealBlue}{rgb}{0.184, 0.496, 0.463}

\usetikzlibrary{arrows,positioning,automata,calc,shapes}

\definecolor{myblue}{rgb}{0.0000,0.4470,0.7410}
\definecolor{myred}{rgb}{0.8500,0.3250,0.0980}
\definecolor{myorange}{rgb}{0.9290,0.6940,0.1250}

\newtheorem{assumption}{Assumption}
\newtheorem{theorem}{Theorem}
\newtheorem{definition}{Definition}

\newtheorem{remark}{Remark}

\newcommand{\farhad}[1]{{\color{black}#1}}
\newcommand{\extra}[1]{{\color{black}#1}}

\DeclareMathOperator*{\argmin}{arg\,min}

\DeclareMathOperator{\sign}{sign}

\begin{document}

\title{\farhad{The Value of Collaboration in Convex Machine Learning with Differential Privacy}}

\author{
\IEEEauthorblockN{
Nan~Wu$^{\ddag}$, 
Farhad~Farokhi$^{*,\dag}$, 
David~Smith$^{*,\S}$, and 
Mohamed~Ali~Kaafar$^{*,\ddag}$}
\IEEEauthorblockA{
$^{\ddag}$Macquarie University \hspace{.15in}
$^{*}$CSIRO's Data61 \hspace{.15in}
$^{\dag}$The University of Melbourne \hspace{.15in}
$^{\S}$Australian National University
}
}

\maketitle

\begin{abstract} In this paper, we apply machine learning to distributed private data owned by multiple data owners, entities with access to non-overlapping training datasets. We use noisy, differentially-private gradients to minimize the fitness cost of the machine learning model using stochastic gradient descent. We quantify the quality of the trained model, using the fitness cost, as a function of privacy budget and size of the distributed datasets to capture the trade-off between privacy and utility in machine learning. This way, we can predict the outcome of collaboration among privacy-aware data owners prior to executing potentially computationally-expensive machine learning algorithms. Particularly, we show that the difference between the fitness of the trained machine learning model using differentially-private gradient queries and the fitness of the trained machine model in the absence of any privacy concerns is inversely proportional to the size of the training datasets squared and the privacy budget squared. We successfully validate the performance prediction with the actual performance of the proposed privacy-aware learning algorithms, applied to: financial datasets for determining interest rates of loans using regression; and detecting credit card frauds using support vector machines.
\end{abstract}


\section{Introduction}
\subsection{Motivation and Contributions}
Data analysis methods using machine learning (ML) can unlock valuable insights for improving revenue or quality-of-service from, potentially proprietary, private datasets. Having large high-quality datasets improves the quality of the trained ML models in terms of the accuracy of predictions on new, potentially untested data. The subsequent improvements in quality can motivate multiple data owners to share and merge their datasets in order to create larger training datasets. For instance, financial institutes may wish to merge their transaction or lending datasets to improve the quality of trained ML models for fraud detection or computing interest rates. However, government regulations (e.g., the roll-out of the General Data Protection Regulation in EU, the California Consumer Privacy Act or the development of the Data Sharing and Release Bill in Australia) increasingly prohibit sharing customer's data without consent~\cite{bennett2018revisiting}. Our work here is motivated by the need to conciliate the tension between quality improvement of trained ML models and the privacy concerns for data sharing.

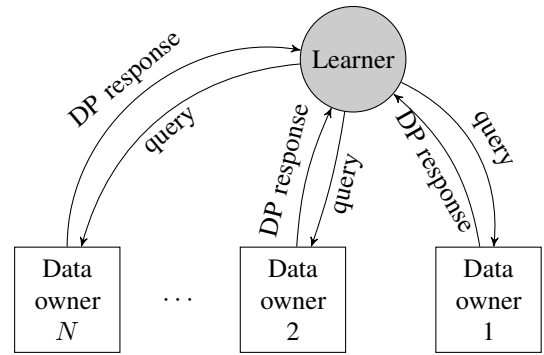
\begin{figure}
\centering
\begin{tikzpicture}[>=stealth']
\node[draw,circle,minimum width=1.4cm,minimum height=1.4cm,fill=black!20] (1) at (-0.5,+1.0) {Learner};
\node[draw,rectangle,minimum width=1.4cm,minimum height=1.4cm] (2) at (+1.3,-2.2) {};
\node[draw,rectangle,minimum width=1.4cm,minimum height=1.4cm] (3) at (-1.3,-2.2) {};
\node[draw,rectangle,minimum width=1.4cm,minimum height=1.4cm] (4) at (-4.3,-2.2) {};
\node[] at (-2.8,-2.2) {$\cdots$};
\draw [->] (1) to [bend left=35] node [sloped,anchor=center,above]  {query} (2);
\draw [->] (2) to [bend right=20] node [sloped,anchor=center,below]  {DP response} (1);
\draw [->] (1) to [bend right=35] node [sloped,anchor=center,below] {query} (4);
\draw [->] (4) to [bend left=50] node [sloped,anchor=center,above] {DP response} (1);
\draw [->] (1) to [bend left=5] node [sloped,anchor=center,below] {query}  (3);
\draw [->] (3) to [bend left=10] node [sloped,anchor=center,above] {DP response} (1);
\node[] at (+1.3,-2.2) {
\begin{minipage}[t]{1cm}
\centering 
Data owner 1
\end{minipage}
};
\node[] at (-1.3,-2.2) {
\begin{minipage}[t]{1cm}
\centering 
Data owner 2
\end{minipage}
};
\node[] at (-4.3,-2.2) {
\begin{minipage}[t]{1cm}
\centering 
Data owner $N$
\end{minipage}
};
\end{tikzpicture}
\caption{\label{fig:0} The communication structure between the learner and the distributed data owners for submitting queries and providing differentially-private (DP) responses.}
\end{figure}

We investigate a machine learning setup in which a learner wants to train a model based on multiple datasets from different data owners. For the purpose of preserving privacy for data contributors, the learner can only submit queries to data owners and they respond by providing differentially-private (DP) responses as illustrated in Figure~\ref{fig:0}. \farhad{We specifically consider honest-but-curious threat models in which different private data owners do not trust each other (or the central learner) for sharing private training datasets, but trust the learner to train the model correctly. As an example, in financial services, a central learner, such as \farhad{a} central bank or government, can be trusted for facilitating computations among banks although they may not trust each other or the learner for accessing private data. Another example is \farhad{for} smart grid in which electricity retailers are private data owners and  \farhad{the} electricity market operator can facilitate learning.} In this paper, the learner submits a gradient query to each data owner. Upon receiving DP responses from data owners to the gradient queries, the learner adjusts the parameters of the ML model in the direction of the average of the DP gradients. Therefore, the quality of the DP responses (in terms of the magnitude of the additive DP noise) from the data owners to the gradient queries determines the performance of the ML training algorithm. 

An important parameter in the ML training algorithm is the step size, the amount by which the model parameters are adjusted in each iteration. If the fitness cost of the ML meets the assumptions of smoothness, strong convexity, and Lipschitz-continuity of the gradient, we can prove that, by selecting the step sizes to be inversely proportional with the iteration number and inversely proportional with the maximum number of iterations squared (see Algorithm~\ref{alg:0} in Section~\ref{sec:DistributedML}), the difference between the fitness of the trained ML model using DP gradient queries and the fitness of the trained ML model in the absence of any privacy concerns becomes small. In fact, the magnitude of the difference becomes inversely proportional to the size of the training datasets squared and the privacy budgets of the data owners squared; see Theorem~\ref{tho:utility_strongconvex} in Section~\ref{sec:theory}. Several ML models and fitness costs, such as linear and logistic regression, satisfy the above-mentioned assumptions. This enables us to predict the outcome of collaboration among privacy-aware data owners and the learner in terms of the fitness cost of the ML training model. However, if the fitness function does not meet these assumptions, we must select the step size to be inversely proportional to the square root of the iteration number. This way, the step size fades away much slower and the effect of the DP noise is more pronounced on the iterates of the learning algorithm. Therefore, we must add an averaging layer on top of the algorithm to reduce the negative impact of the DP noise; see Algorithm~\ref{alg:1} in Section~\ref{sec:DistributedML}. This is based on the developments of~\cite{shamir2013stochastic} with appropriate changes in the averaging step to suit the ML problem with DP gradient queries. In this case, we can prove that the difference between the fitness of the trained ML model using DP gradient queries and the fitness of the trained ML model in the absence of any privacy concerns is inversely proportional to the size of the training datasets (no longer squared) and the  privacy budget (no longer squared); see Theorem~\ref{tho:2} in Section~\ref{sec:theory}. 

\farhad{In this paper, we focus on the case \extra{where} the datasets in possession of the private data owners in Figure~\ref{fig:0} are mutually exclusive or non-overlapping, i.e., two identical records are not shared across the datasets. In many real-life applications within \farhad{the} financial and energy sector\farhad{s}, this is a realistic assumption, e.g, transactional records (e.g. for purchasing goods) are unique by the virtue of timestamps, amounts, and the uniqueness of purchases \farhad{by} an individual. This assumption is set in place to ensure differential privacy using independent additive noises. In the absence of such an assumption, there also needs to be a privacy-preserving mechanism for identifying those common entries without potential information leakage regarding non-common entries, which itself is a daunting task and open problem for research.} 

For experimental verification of the theoretical results, two financial datasets are used in this paper. First, we use a regression model on a dataset containing information on loans made on Lending Club, a peer-to-peer lending platform~\cite{kaggle1}, to automate the process of setting interest rates of loans. Second, we train a support vector machine for detecting fraudulent transactions based on a dataset containing transactions made by European credit card-holders in September 2013~\cite{kaggle2}. We use the experiments to validate theoretical predictions and to gain important insights into the outcome of collaborations among privacy-aware data owners. For instance, even if the learner has access to one large dataset with relaxed privacy constraints, the performance of the trained ML model can be very bad if small \farhad{conservative datasets (i.e., datasets with} very small privacy budgets) also contribute to the learning. Therefore, it is best to exclude smaller conservative datasets from collaboration. This is a \textit{counter-intuitive observation as it clearly indicates that more data is not always good}, if it is obfuscated by conservative data owners. Larger, but conservative, datasets are sometimes worth including in the training as they do not degrade performance heavily with their conservative privacy budgets, yet improve the performance of the trained ML model because of their size. \farhad{These observations can be alternatively interpreted as\farhad{:} collaboration in training a model with a dataset can only be useful if and only if it has enough information (i.e., enough data entries) to suppress the impact of random noise added for privacy guarantees.}

In summary, this paper makes the following contributions:
\begin{itemize}
\item We develop DP gradient descent algorithms for training ML models on distributed private datasets owned by different entities; see Algorithms~\ref{alg:0} and~\ref{alg:1} in Section~\ref{sec:DistributedML}.
\item We prove that the quality of the trained ML model using DP gradient descent algorithm scales inversely with privacy budgets squared, and the size of the distributed datasets squared, which can establish a trade-off between privacy and utility in privacy-preserving ML; 
\item We develop a theory that enables to predict the outcome of a potential collaboration among privacy-aware data owners (or data custodians) in terms of the fitness cost of the ML training model prior to executing potentially computationally-expensive ML algorithms on distributed privately-owned datasets; see Theorems~\ref{tho:utility_strongconvex} and~\ref{tho:2} in Section~\ref{sec:theory}. \farhad{The bounds in these theorems are not necessarily optimal, i.e., there might exist better performance bounds for other privacy-preserving learning algorithms, but, if the data owners follow Algorithms~\ref{alg:0} and~\ref{alg:1}, they can predict their success or failure.}
\item We validate our theoretical analysis by evaluating our differentially private ML algorithms using distributed \farhad{non-overlapping} financial datasets belonging to multiple institutes/banks for determining interest rates of loans using regression, and for  detecting credit card fraud using support vector machine classifier; We further validate the predictions of the analysis with the actual performance of the proposed privacy-aware learning algorithms applied to the distributed financial datasets; see Section~\ref{sec:numerical}.
\item Our experimental results indicate that, in the case of three banks collaborating to train a support vector machine classifier to detect credit card fraud, within only 100 iterations, the fitness of the trained model using DP gradient queries is in average within 90\% of the fitness of the trained model in the absence of privacy concern if the  privacy budget is equal to 1 and each bank has access to a dataset of 30,000 records of credit card transactions and their validity. We observe similar performance results for training a regression model over interest rates of loans with the privacy budget of 10 and datasets of 350,000 records each.
\end{itemize}



\subsection{Related Work}
\textbf{ML using Secure Multi-Party Computation and Encryption.} Secure multi-party computation provide avenues for securing the iterations of distributed ML  algorithms across multiple data owners. In the past, secure multi-party computation has been used in various ML models, such as decision trees \cite{Lindell1010073540445986_3}, regression~\cite{du2004privacy}, association rules~\cite{vaidya2002privacy}, and clustering~\cite{vaidya2008privacy,jagannathan2005privacy}. Training ML models using encrypted data was discussed in~\cite{bonawitz2017practical,graepel2012ml, hunt2018chiron,li2017multi,aono2018privacy}. In~\cite{gilad2016cryptonets}, efficient conversion of models for use of encrypted input data was discussed. The use of secure multi-party computation reduces the computational efficiency of ML algorithms by adding a non-trivial computational and communication performance overhead.

\textbf{ML with Differential Privacy.}
A natural way for alleviating privacy concerns is to deploy privacy-enabled ML using differential privacy (DP)~\cite{sarwate2013signal,zhang2012functional, chaudhuri2009privacy,zhang2016differential}. In~\cite{chaudhuri2009privacy}, a privacy-preserving regularized logistic regression algorithm is provided for learning from private databases by bounding the sensitivity of regularized logistic regression, and perturbing the learned classifier with noise proportional to the sensitivity. This technique is proved to be DP and simulations are used to investigate the trade-off between privacy and learning utility. In~\cite{zhang2012functional}, a large class of optimization-based DP machine learning algorithms are developed by appropriately perturbing the objective function of the ML training algorithm. The mechanism is applied to linear and logistic regression models and shown to provide high accuracy. In the mentioned studies, privacy-preserving ML, however, often relies on an entire dataset, constructed by merging smaller datasets, being stored in one location. The ML model is then either trained on the aggregated dataset, and is systematically obfuscated using additive noise to guarantee differential privacy, or trained on an obfuscated centrally-located data. Such methods do not address the underlying problem that the smaller datasets are owned by multiple entities with restrictions on sharing sensitive data.

\textbf{Distributed/Collaborative Privacy-Preserving ML.}
ML based on distributed private datasets has been recently investigated in, e.g., ~\cite{zhang2017dynamic,huang2018dp,  huang2015differentially,nozari2018differentially, hale2015differentially}. Note that this problem is intimately related to distributed optimization using differentially-private  oracles, as such ML problems can be cast as  distributed optimization problems in which distributed training datasets are represented within cost functions or constraints of the entities. 
Using stochastic gradient descent with additive Gaussian/Laplace noise to ensure DP is also common in the literature; (e.g.,~\cite{shokri2015privacy,abadi2016deep,mcmahan2017learning,zhang2018privacy}). In~\cite{shokri2015privacy}, noisy gradients are used to train a deep neural network. The scale of the required additive noise for DP is reduced in~\cite{abadi2016deep} by employing the idea of moment accountant, instead of standard composition rules. Stochastic gradient descent is also utilized in~\cite{mcmahan2017learning} for recurrent neural network language models. Generalizations for obfuscating individual and group-level trends by DP additive noise are presented in~\cite{zhang2018privacy}. 
Because iterative methods rely on multiple rounds of inquiries of private datasets, for instance, by submitting multiple gradient queries, the privacy budget must be inversely scaled by the total number of iterations to ensure that a reasonable privacy guarantee can be achieved (alternatively, privacy guarantees get weaker as the number of iterations grows because of the composition rule of differential privacy). Hence, if the parameters of the optimization algorithm are not carefully chosen, bounds on the performance of the ML training algorithm deteriorates with an increasing total number of iterations; e.g., see~\cite{han2017differentially}. In~\cite{zhang2017dynamic,huang2018dp}, the  privacy budget was kept constant and therefore by communicating more, as the number of the iterations grows, the privacy guarantee weakens. However, in those studies, if the privacy budget had been scaled inversely proportional to the total number of iterations, privacy guarantees would be maintained over the entire horizon but performance would deteriorate with increasing total number of iterations, as in~\cite{han2017differentially}. 

All these studies, however, do not address the issues of convergence of the learning algorithm, selection of appropriate step size in the stochastic gradient descent, and forecasting of the quality of the trained ML model based on the privacy budget prior to running extensive potentially  computationally-expensive experiments. These missing steps are some of the important contributions of this paper.

\subsection{Paper Organization}
The rest of the paper is organized as follows. We introduce our system model and propose privacy-aware ML algorithms with distributed private datasets in Section~\ref{sec:DistributedML}. We analyze and provide theoretical results for predicting the performance of the privacy-preserving training algorithms in Section~\ref{sec:theory}. We present the experimental results in Section~\ref{sec:numerical}. Finally, we conclude the paper in Section~\ref{sec:conclusions}.

\section{ML Training Algorithm Based on Distributed Private Data with DP Gradient Queries}\label{sec:DistributedML}
\subsection{Setup} \label{subsec:motivation}
Consider a group of $N\in\mathbb{N}$ private agents or data owners $\mathcal{N}:=\{1,\dots,N\}$ that are connected to a node responsible for training a ML model, identified as a learning agent, over an undirected communication graph as in Figure~\ref{fig:0}. Each agent has access to a set of private training data $\mathcal{D}_i:=\{(x_i,y_i)\}_{i=1}^{n_i}\subseteq \mathbb{X}\times\mathbb{Y}\subseteq\mathbb{R}^{p_x}\times\mathbb{R}^{p_y}$, where $x_i$ and $y_i$, respectively, denote inputs and  outputs. Each data owner, for instance, could be a private bank/financial institution. In this case, the private datasets can represent information about loan applicants (such as salary, employment status, and credit rating\footnote{Categorical attributes, such as gender, can always be translated into numerical ones according to a rule.}) as inputs and historically approved interest rates per annum by the bank (in percentage points) as outputs. 

\begin{assumption} \label{assum:exclusive} Private datasets are mutually exclusive, i.e., $\mathcal{D}_i\cap\mathcal{D}_j=\emptyset$ for all $i,j\in\mathcal{N}$.
\end{assumption}

\farhad{
Assumption~\ref{assum:exclusive} states that two identical records, equal in every possible aspect, cannot be in two or more datasets. This is a realistic assumption in many real-life applications, such as financial and energy data. For instance, across multiple banks and financial-service providers, transaction records (e.g. for purchasing goods) are unique by the virtue of timestamps, amounts, and the uniqueness of purchases for an individual. In energy systems, one household cannot transact (for purchasing power) with two or more energy retailers and thus its consumption patter\farhad{n} can only be stored \farhad{by} one retailer. The reasons behind this assumption are two\farhad{-}fold. First, to guarantee $\epsilon$-differential privacy, we need to ensure that the records are not repeated so that an adversary cannot reduce the noise levels by averaging the reports containing information about repeated entries and thus \farhad{exceeding} $\epsilon$ (due to \farhad{the} composition rule for differential privacy). If the datasets \farhad{had} common entries, there \farhad{would} need to be a privacy-preserving mechanism for identifying those common entries without potential information leakage \farhad{with respect to} non-common entries, which is a daunting task. The mutually exclusive or non-overlapping nature of the datasets also results in statistical independence of additive privacy-preserving noise. This independence is extremely useful in comput\farhad{ing} the magnitude of the additive noise for forecasting the performance of privacy-aware learning algorithms. \extra{If records can appear in at most $\kappa\in\{1,\dots, N\}$ datasets and we do not exclude the overlapping entries during the learning, we must ensure that the gradient queries are DP with privacy budget $\epsilon_i/\kappa, \forall i\in\mathcal{N}$. This is to ensure that we can guarantee privacy budget $\epsilon_i$ for the repeated entries across the datasets by using the composition rule for differential privacy. This results in degradation of the fitness of the trained ML model with privacy-preserving algorithms. For instance, in Theorem~\ref{tho:utility_strongconvex}, we show that the difference between the fitness of the trained ML model using DP gradient queries and the fitness of the trained ML model in the absence of any privacy concerns is inversely proportional to the size of the training datasets squared and the privacy budget squared. Therefore, when allowing repeated entries, the difference between the fitness of the private ML model and the fitness of the trained machine model without privacy concerns degrades by a factor of $\kappa^2$.}}

The learning agent is interested in extracting a meaningful relationship between the inputs and outputs using ML model $\mathfrak{M}:\mathbb{X}\times\mathbb{R}^{p_\theta}\rightarrow\mathbb{Y}$ and the available training datasets $\mathcal{D}_i$, $\forall i\in\mathcal{N}$, by solving the optimization problem in
\begin{align}\label{eqn:ML}
\theta^*\in\argmin_{\theta\in\Theta}\Bigg[g_1(\theta)+\frac{1}{n} \sum_{\farhad{j}\in\mathcal{N}}\sum_{\{x,y\}\in\mathcal{D}_j}\hspace{-.1in} g_2(\mathfrak{M}(x;\theta),y)\Bigg],
\end{align}
where $g_2(\mathfrak{M}(x;\theta),y)$ is a loss function capturing the ``closeness'' of the outcome of the trained ML model $\mathfrak{M}(x;\theta)$ to the actual output $y$,  $g_1(\theta)$ is a regularizing term, $n:=\sum_{\ell\in\mathcal{N}}n_{\ell}$, and $\Theta:=\{\theta\in \mathbb{R}^{p_\theta}\,|\,\| \theta \|_\infty\leq \theta_{\max}\}.$ Note that a large enough $\theta_{\max}$ can always be selected such that the search over $\Theta$ does not add any conservatism (in comparison to the unconstrained case), if desired. We use $f(\theta)$ to denote the cost function of~\eqref{eqn:ML} for the sake of the brevity of the presentation, i.e., 
\begin{align}
f(\theta):=g_1(\theta)+\frac{1}{n}\sum_{\{x,y\}\in\bigcup_{j\in\mathcal{N}}\mathcal{D}_j}\hspace{-.1in} g_2(\mathfrak{M}(x;\theta),y).
\end{align}  

\begin{remark}[Generality of Optimization-Based ML] In an automated loan assessment example, a bank maybe  interested in  employing a linear regression model to estimate the interest rate of the loans based on attributes of customers (thus developing an ``AI platform'' for loan assessment and delivery). A linear regression model, as the name suggests, considers a linear relationship between input $x$ and output $y$ in the form of $y=\mathfrak{M}(x;\theta):=\theta^\top x$, where $\theta\in\mathbb{R}^{p_\theta}$ is the parameter of the ML model. We can train the regression model by solving the optimization problem~\eqref{eqn:ML} with $g_2(\mathfrak{M}(x;\theta),y)=\|y-\mathfrak{M}(x;\theta)\|_2^2$, and $g_1(\theta)=0$. In addition to linear (or non-linear) regression discussed earlier, which clearly is of the form in~\eqref{eqn:ML}, several other ML algorithms follow this formulation. Another example is linear support vector machines (L-SVM). In this problem, it is desired to obtain a separating hyper plane of the form $\{x\in\mathbb{R}^{p_x}: \theta^\top [x^\top \; 1 ]^\top=0\}$ with its corresponding classification rule $\sign(\mathfrak{M}(x;\theta))$ with $\mathfrak{M}(x;\theta):=\theta^\top [x^\top \; 1 ]^\top$ to group the training data into two sets (corresponding to $y=+1$ and $y=-1$).  This problem can be cast as~\eqref{eqn:ML} with $g_1(\theta):=(1/2)\theta^\top \theta $ and $ g_2(\mathfrak{M}(x;\theta),y):=\max(0,1-\mathfrak{M}(x;\theta)y).$ We can easily see that the extension to non-linear SVM can also be cast as an optimization-based ML problem. Another example is artificial neural network (ANN). In this case, $\mathfrak{M}(x;\theta)$ describes the input-output behaviour of the ANN with $\theta$ capturing parameters, such as internal thresholds. This problem can be cast as~\eqref{eqn:ML} with $g_1(\theta):=0$ and $g_2(\mathfrak{M}(x;\theta),y):=\|y-\mathfrak{M}(x;\theta))\|_2.$
\end{remark}

If the data owners could come to an agreement to share private data (and it was not illegal to disclose customers' private information without their consent), the learning agent could train the ML model by solving the optimization problem~\eqref{eqn:ML} directly. In practice, however, data owners may not be able to share their private data. In this case, the learning agent can submit queries $\mathfrak{Q}_i(\mathcal{D}_i;k)\in\mathcal{Q}$ to agent $i\in\mathcal{N}$ for $k\in\mathcal{T}:=\{1,\dots,T\}$, where $T$ denotes the number of communication rounds (i.e., the number of queries) agreed upon by all the data owners prior to the exchange of information, index $k$ identifies the current communication round, and $\mathcal{Q}$ denotes the output space of the query. Agent $i\in\mathcal{N}$ can then provide a differentially-private response $\overline{\mathfrak{Q}}_i(\mathcal{D}_i;k)\in\mathcal{Q}$ to the query $\mathfrak{Q}_i(\mathcal{D}_i;k)\in\mathcal{Q}$. 

\begin{definition}[Differential Privacy] \label{def:dp} The response policy of data owner $\ell\in\mathcal{N}$  is  $\epsilon_\ell$-differentially private over the horizon $T$ if
\begin{align*}
\mathbb{P}\bigg\{(\overline{\mathfrak{Q}}_\ell(\mathcal{D}_\ell;k&))_{k=1}^T\in\mathcal{Y}\bigg\}\\
&\leq \exp(\epsilon_\ell)\mathbb{P}\bigg\{(\overline{\mathfrak{Q}}_\ell(\mathcal{D}'_\ell;k))_{k=1}^T\in\mathcal{Y}\bigg\},
\end{align*}
where $\mathcal{Y}$ is any Borel-measurable subset of $\mathcal{Q}^T$, and $\mathcal{D}_\ell$ and $\mathcal{D}'_\ell$ are two adjacent datasets differing at most in one entry, i.e., $|\mathcal{D}_\ell\setminus\mathcal{D}'_\ell|=|\mathcal{D}'_\ell\setminus\mathcal{D}_\ell|\leq 1$.
\end{definition}

The learning agent then processes all the received responses to the queries in order to generate its ML model:
\begin{align*}
\hat{\theta}:=\varsigma((\overline{\mathfrak{Q}}_j(\mathcal{D}_j;k))_{k\in\mathcal{T},j\in\mathcal{N}}),
\end{align*}
where $\varsigma:\prod_{k\in\mathcal{T}}\mathcal{Q}^T\rightarrow\mathbb{R}^{p_\theta}$ is a mapping used by the learning agent for fusing all the available information. 

In the next subsection, we present an algorithm for generating queries, and then use the provided differentially-private responses for computing a trained ML model.  

\subsection{Algorithm} \label{sec:ML_DP}
In the absence of privacy concerns, one strategy for training the ML model by the learning agent is to provide unfettered access to the original private data of the data owners in $\mathcal{N}$. In this case, the learning agent can follow the projected (sub)gradient descent iterations in
\begin{align} \label{eqn:project_subgradient_0}
\theta[k+1]
&=
\Pi_{\Theta}[
\theta[k]
-
\rho_k
\xi_f(\theta[k])],
\end{align}
where $\rho_k>0$ is the step-size at iteration $k$, $\xi_f(\theta[k])$ is a sub-gradient, an element of sub-differentials $\partial_\theta f(\theta[k])$, of the cost function $f$ with respect to the variable $\theta$ evaluated at $\theta[k]$~\cite{shor2012minimization}, and  $\Pi_\Theta[\cdot]$ denotes projection operator into the set $\Theta$ defined as $\Pi_\Theta[a]:=\argmin_{b\in\Theta}\|a-b\|_2.$ For  continuously differentiable functions, the gradient is the only sub-gradient. The use of sub-gradients, instead of gradient in this paper, is motivated by the possible choice of non-differentiable loss functions in ML, e.g., the cost function of the L-SVM.

\begin{assumption} \label{assum:convex} $g_1$ and $g_2$ are convex functions of $\theta$. 
\end{assumption}

Assumption~\ref{assum:convex} implies that $f$ is also a convex function of~$\theta$. The existence of sub-differentials is guaranteed for convex functions~\cite{shor2012minimization}. We define $\bar{g}_2^{x,y}(\theta)=g_2(\mathfrak{M}(x;\theta),y)$.  The update law in~\farhad{\eqref{eqn:project_subgradient_0}} can be rewritten as 
\begin{align}
\theta[k+1]
=
\Pi_{\Theta}\Bigg[&
\theta[k]
-
\rho_k
\xi_{g_1}(\theta[k]) \nonumber \\&-\frac{\rho_k}{n} \sum_{\ell\in\mathcal{N}_j}\sum_{\{x,y\}\in\mathcal{D}_\ell} \xi_{\bar{g}_2^{x,y}}(\theta[k])\Bigg],\nonumber\\
=
\Pi_{\Theta}\Bigg[&
\theta[k]-
\rho_k
\xi_{g_1}(\theta[k])\nonumber
\\
&-\frac{\rho_k}{n} \sum_{\ell\in\mathcal{N}_j\setminus\{j\}} n_\ell \mathfrak{Q}_\ell(\mathcal{D}_\ell;k)\Bigg],
\end{align}
where $\xi_{g_1}$ is a sub-gradient of $g_1$, $\xi_{\bar{g}_2^{x,y}}$ is a sub-gradient of $\bar{g}_2^{x,y}$, and $\mathfrak{Q}_\ell(\mathcal{D}_\ell;k)$ is a query that can be submitted by the learning agent to data owner $\ell\in\mathcal{N}$ in order to provide the aggregate sub-gradient:
\begin{align}
\mathfrak{Q}_\ell(\mathcal{D}_\ell;k)=\frac{1}{n_\ell}\sum_{\{x,y\}\in\mathcal{D}_\ell} \xi_{\bar{g}_2^{x,y}}(\theta[k]).
\end{align}
Responding to the query $\mathfrak{Q}_\ell(\mathcal{D}_\ell;k)$ clearly intrudes on the privacy of the individuals in dataset $\mathcal{D}_\ell$. Therefore, data owner $\ell$ only responds in a differentially-private manner by reporting the noisy aggregate:
\begin{align} \label{eqn:diff_privacy_additive_noise}
\overline{\mathfrak{Q}}_\ell(\mathcal{D}_\ell;k)
=\mathfrak{Q}_\ell(\mathcal{D}_\ell;k)+w_\ell[k],
\end{align}
where $w_\ell[k]$ is an additive noise to establish differential privacy with privacy budget $\epsilon_\ell$ over the horizon $T$; see Definition~\ref{def:dp}. 
As stated before, here, the horizon $T$ is the total number of iterations of the projected sub-gradient algorithm. Note that each neighbour responds to one query in each iteration.

\begin{assumption} \label{assum:maximugradient}
$\Xi:=\max_{(x,y)\in\mathbb{X}\times\mathbb{Y}}\| \xi_{\bar{g}_2^{x,y}}(\theta[k])\big\|_1<\infty$. 
\end{assumption}

Assumption~\ref{assum:maximugradient} implies the gradients or the sub-gradients of fitness function have a bounded magnitude. \farhad{For strongly  convexity loss functions with Lipschitz gradients, this assumption can be satisfied. This is because, for strongly convex functions, the decision variables, i.e., the ML model, remains within a compact set. However, for no\farhad{n-}strongly convex functions, we need to restrict the ML models to  \farhad{the} compact set $\Theta$; see~\eqref{eqn:ML}.}

\begin{theorem} \label{tho:1} The policy of data owner $\ell$ in~\eqref{eqn:diff_privacy_additive_noise}  for responding to the queries is  $\epsilon_\ell$-differentially private over horizon $\{1,\dots,T\}$ if $w_\ell[k]$ are i.i.d.\footnote{independently and identically distributed} noises with the density function
\begin{align*}
p(w)=\bigg(\frac{1}{2b}\bigg)^{p_\theta}\exp\bigg(-\frac{\|w\|_1}{b} \bigg)
\end{align*}
with  scale $b=2\Xi T/(n_\ell \epsilon_\ell)$. 
\end{theorem}

\begin{proof} See Appendix~\ref{proof:tho:1}.
\end{proof}

\begin{algorithm}[t]
\caption{\label{alg:0} ML training algorithm with distributed private datasets using DP gradients for strongly-convex smooth fitness cost. }
\begin{algorithmic}[1]
\REQUIRE $T$
\ENSURE $(\theta[k])_{k=1}^T$
\STATE Initialize $\theta[1]$
\FOR{$k=1,\dots,T-1$}
\STATE Learner submits query $\mathfrak{Q}_\ell(\mathcal{D}_\ell;k)$ to data owners in $\mathcal{N}$
\STATE Data owners return DP responses $\overline{\mathfrak{Q}}_\ell(\mathcal{D}_\ell;k)$ 
\STATE Learner follows the update rule
\begin{align*}
\hspace{-.1in}\theta[k+1]=\theta[k]&-\frac{\rho}{T^2k}\bigg(\xi_{g_1}(\theta[k])+\sum_{\ell\in\mathcal{N}}\frac{n_\ell}{n}\overline{\mathfrak{Q}}_\ell(\mathcal{D}_\ell;k)\bigg),
\end{align*}
\ENDFOR
\end{algorithmic}
\end{algorithm}

\farhad{Theorem~\ref{tho:1} states that i.i.d.  Laplace additive noise can ensure DP gradients.}
Each response in~\eqref{eqn:diff_privacy_additive_noise}, for a given $k$, using the additive noise density in Theorem~\ref{tho:1} is $(\epsilon_\ell/T)$-differentially private. Therefore, over the whole horizon $\{1,\dots,T\}$, all the responses meet the definition of $\epsilon_\ell$-differential privacy. This follows from the composition of $T$ differentially-private mechanisms~\cite{dwork2014algorithmic}. In~\cite{zhang2017dynamic, huang2018dp}, each response is constructed to ensure $\epsilon$-differential privacy, which implies that the overall algorithm is $\epsilon T$-differentially private, thus reducing the privacy guarantee with increasing the number of the iterations. 

In the presence of the additive noise, the iterates of  the learner follow the stochastic map 
\begin{align} \label{eqn:noisy_subgradient}
\theta[k+1]
&=
\Pi_{\Theta}[
\theta[k]
-
\rho_k(
\xi_f(\theta[k])+w[k])],
\end{align}
where
\begin{align*}
w[k]:=\frac{1}{n}\sum_{\ell\in\mathcal{N}}n_\ell w_\ell[k].
\end{align*}
Algorithm~\ref{alg:0} summarizes our proposed ML algorithm with distributed private datasets using DP gradients. \farhad{Note that, in Algorithm~\ref{alg:0}, the step size, or the learning rate, \extra{decreases} with the iteration number $k$. This is done to reduce the influence of the privacy-preserving additive noise in the performance of the trained model. In the non-private training (i.e., when $\epsilon=+\infty$), we do not need to reduce the step size with iteration number $k$ as there is no privacy-preserving noise. In fact, we can select a constant learning rate to extract the non-private  model; see~\cite{grimmer2019convergence} for convergence analysis of optimization algorithms with constant steps sizes.}

In Section~\ref{sec:theory}, we observe that the performance of Algorithm~\ref{alg:0} can only be assessed under the assumptions of differentiability, smoothness, and strong convexity of the fitness cost. These assumptions are satisfied for several ML models and fitness costs, such as regression. To avoid these assumptions and to also reduce the effect of the additive noise, we can define the averaging variable
\begin{align}
\bar{\theta}[k+1]
&=
\bigg(1-\frac{1/\sqrt{T}+1}{1/\sqrt{T}+k} \bigg)
\bar{\theta}[k]
+\frac{1/\sqrt{T}+1}{1/\sqrt{T}+k} 
\theta[k]\nonumber\\
&=
\frac{k-1}{1/\sqrt{T}+k} 
\bar{\theta}[k]
+\frac{1/\sqrt{T}+1}{1/\sqrt{T}+k} 
\theta[k].\label{eqn:noisy_subgradient_averaged}
\end{align}
Algorithm~\ref{alg:1} summarizes the proposed ML algorithm with distributed private datasets using DP sub-gradients with the additional averaging step as per equation~\eqref{eqn:noisy_subgradient_averaged}. Now, we are ready to analyze the performance our privacy-preserving ML training algorithms.

\begin{algorithm}[t]
\caption{\label{alg:1} ML algorithm with distributed private datasets using DP sub-gradients. }
\begin{algorithmic}[1]
\REQUIRE $T$, $c_1$
\ENSURE $(\theta[k])_{k=1}^T$
\STATE Initialize $\theta[1]$ within $\Theta$
\FOR{$k=1,\dots,T-1$}
\STATE Learner submits query $\mathfrak{Q}_\ell(\mathcal{D}_\ell;k)$ to data owners in $\mathcal{N}$
\STATE Data owners return DP responses $\overline{\mathfrak{Q}}_\ell(\mathcal{D}_\ell;k)$ 
\STATE Learner follows the update rule
\begin{align*}
\hspace{-.3in}\theta[k+1]=\Pi_{\Theta}\bigg[\theta[k]&-\frac{c_1}{\sqrt{k}}\bigg(\hspace{-.03in}\xi_{g_1}(\theta[k])\hspace{-.03in}+\hspace{-.04in}\sum_{\ell\in\mathcal{N}}\frac{n_\ell}{n}\overline{\mathfrak{Q}}_\ell(\mathcal{D}_\ell;k)\hspace{-.04in}\bigg)\bigg],
\end{align*}
\STATE Learner follows the averaging rule
\begin{align*}
\bar{\theta}[k+1]
&=
\frac{k-1}{1/\sqrt{T}+k} 
\bar{\theta}[k]
+\frac{1/\sqrt{T}+1}{1/\sqrt{T}+k} 
\theta[k].
\end{align*}
\ENDFOR
\end{algorithmic}
\end{algorithm}

\section{Predicting the Performance of ML on Distributed Private Data} \label{sec:theory}

For Algorithm~\ref{alg:0}, we can prove the following convergence result under the assumptions of  differentiability, smoothness, and strong convexity of the ML fitness function.

\begin{theorem} \label{tho:utility_strongconvex} Assume that $f$ is a $L$-strongly convex continuously-differentiable function with $\lambda$-Lipschitz gradient and $\theta_{\max}=\infty$ (i.e., there is no constraint). For any $\varepsilon>0$, there exists a large enough $T$ such that the iterates of Algorithm~\ref{alg:0} satisfy
\begin{align} \label{tho:utility_strongconvex:eq1}
\min_{1\leq k\leq T} \mathbb{E}\{f(\theta[k])\}-f(\theta^*)
\leq & \frac{8\Xi^2 \rho}{L n^2}\bigg(\sum_{\ell\in\mathcal{N}}\frac{1}{\epsilon_\ell^2}\bigg)+\varepsilon,
\end{align}
and
\begin{align}\label{tho:utility_strongconvex:eq2}
\min_{1\leq k\leq T}\mathbb{E}\{\|\theta[k]-\theta^*\|_2^2\}\leq & \frac{32\Xi^2 \rho}{L^2 n^2}\bigg(\sum_{\ell\in\mathcal{N}}\frac{1}{\epsilon_\ell^2}\bigg)+\frac{\varepsilon}{4L}.
\end{align}
\end{theorem}

\begin{proof} See Appendix~\ref{proof:tho:utility_strongconvex}.
\end{proof}

Theorem~\ref{tho:utility_strongconvex} establishes the convergence of Algorithm~\ref{alg:0} for smooth strongly convex functions. This quantifies the \textit{trade-off between privacy and utility} by capturing the closeness to the trained ML model with and without taking into account the privacy constraints of the data owners. In fact, the inequalities in~\eqref{tho:utility_strongconvex:eq1} and~\eqref{tho:utility_strongconvex:eq2} enable us to predict the outcome of a potential collaboration among privacy-aware data owners (or data custodians) in terms of the fitness cost of the ML training model prior to executing potentially computationally-expensive ML algorithms on distributed privately-owned datasets.

To relax the conditions required for convergence of the ML training, we can use Algorithm~\ref{alg:1}. In this case, we do not even need the fitness function to be differentiable because the algorithm uses  sub-gradients, rather than gradients. For the noisy projected sub-gradient decent algorithm in Algorithm~\ref{alg:1}, the following result can be proved.

\begin{theorem} \label{tho:2} For any $T$, there exists large enough constants\footnote{Note that the constants in the statement of the theorem can be functions of $T$ and, therefore, the bounds in~\eqref{eqn:1} and~\eqref{eqn:2} are useful for comparing the variations in the performance of the sub-gradient descent algorithm for various privacy budgets and sizes of the datasets as long as $T$ is fixed.  } $c_1,c_2>0$ such that the iterates of Algorithm~\ref{alg:1} satisfy
\begin{align} \label{eqn:1}
\mathbb{E}\{f(\bar{\theta}[T])\}-f(\theta^*)&\leq  \frac{c_2\Xi}{n}\sqrt{\sum_{\ell\in\mathcal{N}}\frac{1}{\epsilon_\ell^2}},
\end{align}
Further, if $g_1$ is a $L$-strongly convex function, 
\begin{align}
 \label{eqn:2}
\mathbb{E}\bigg\{
\big\|\bar{\theta}[T]-\theta^*\big\|_2^2\bigg\}&\leq 
\frac{4c_2\Xi}{Ln }\sqrt{\sum_{\ell\in\mathcal{N}}\frac{1}{\epsilon_\ell^2}}.
\end{align}
\end{theorem}

\begin{proof} See Appendix~\ref{proof:tho:2}.
\end{proof}

The upper bounds on the performance of the training Algorithms~\ref{alg:0} and~\ref{alg:1} in Theorems~\ref{tho:utility_strongconvex} and~\ref{tho:2} are increasing functions of  $(1/n^2)\sum_{\ell\in\mathcal{N}}1/(\epsilon_\ell)^2$ and $(1/n)[\sum_{\ell\in\mathcal{N}}1/(\epsilon_\ell)^2]^{1/2}$, respectively. By increasing $\epsilon_\ell$, i.e., relaxing the privacy guarantees of data owners, the performance of the ML training algorithm improves, as expected because of having access to better quality gradient oracles. 

\farhad{
\begin{remark}[Comparison with Central Bounds]
Under the assumption that all the data owners have equal privacy budgets $\epsilon_i=\epsilon$, $\forall i$, the bound in~\eqref{tho:utility_strongconvex:eq1} scales as $\epsilon^{-2}$ and the bound in~\eqref{eqn:1} scales as $\epsilon^{-1}$. These bounds are in line with the lower and the upper bounds in~\cite{bassily2014private} for strongly convex and general convex loss functions. The same outcome also holds if $N=1$ and $\epsilon_1=\epsilon$, which is the case of centralized privacy-preserving learning. 
\end{remark}
}

\farhad{Finally, we note that these results provide bounds on the distance between the non-private ML model and the privacy-preserving ML models learned in a distributed manner as a function of the privacy budgets and the size of the datasets. Issues, such as non-independent and non-identical datasets, influence the performance of the non-private model and thus also indirectly influence the performance of the privacy-preserving models. In the next section, although the datasets are not restricted be i.i.d. (e.g., the number of fraudulent transactions in the credit card fraud detection is low and arguably contains activities \farhad{that have} originated from same/similar fraudsters), the theoretical bounds tightly match the experimental results.}

\begin{figure*}
\centering
\begin{tabular}{ccc}
\hspace{-.25in}
$\epsilon_1=\epsilon_2=\epsilon_3=0.1$
\hspace{-.3in}
&
\hspace{-.2in}
$\epsilon_1=\epsilon_2=\epsilon_3=1.0$
\hspace{-.3in}
&
\hspace{-.2in}
$\epsilon_1=\epsilon_2=\epsilon_3=10$
\hspace{-.2in}
\\[-.5em]
\hspace{-.25in}
\begin{tikzpicture}
\node[] at (0,0) {
\psfrag{x1}[cc][][0.6][0]{\quad\quad$10^0$}
\psfrag{x2}[cc][][0.6][0]{$10^{1}$}
\psfrag{x3}[cc][][0.6][0]{$10^{2}$}
\psfrag{yyy1}[cl][bl][0.6][45]{$10^{-1}$}
\psfrag{yyy2}[cc][][0.6][45]{$10^{0}$}
\psfrag{yyy3}[cc][][0.6][45]{$10^{1}$}
\psfrag{yyy4}[cc][][0.6][45]{$10^{2}$}
\psfrag{yyy5}[cc][][0.6][45]{$10^{3}$}
\includegraphics[width=.35\linewidth]{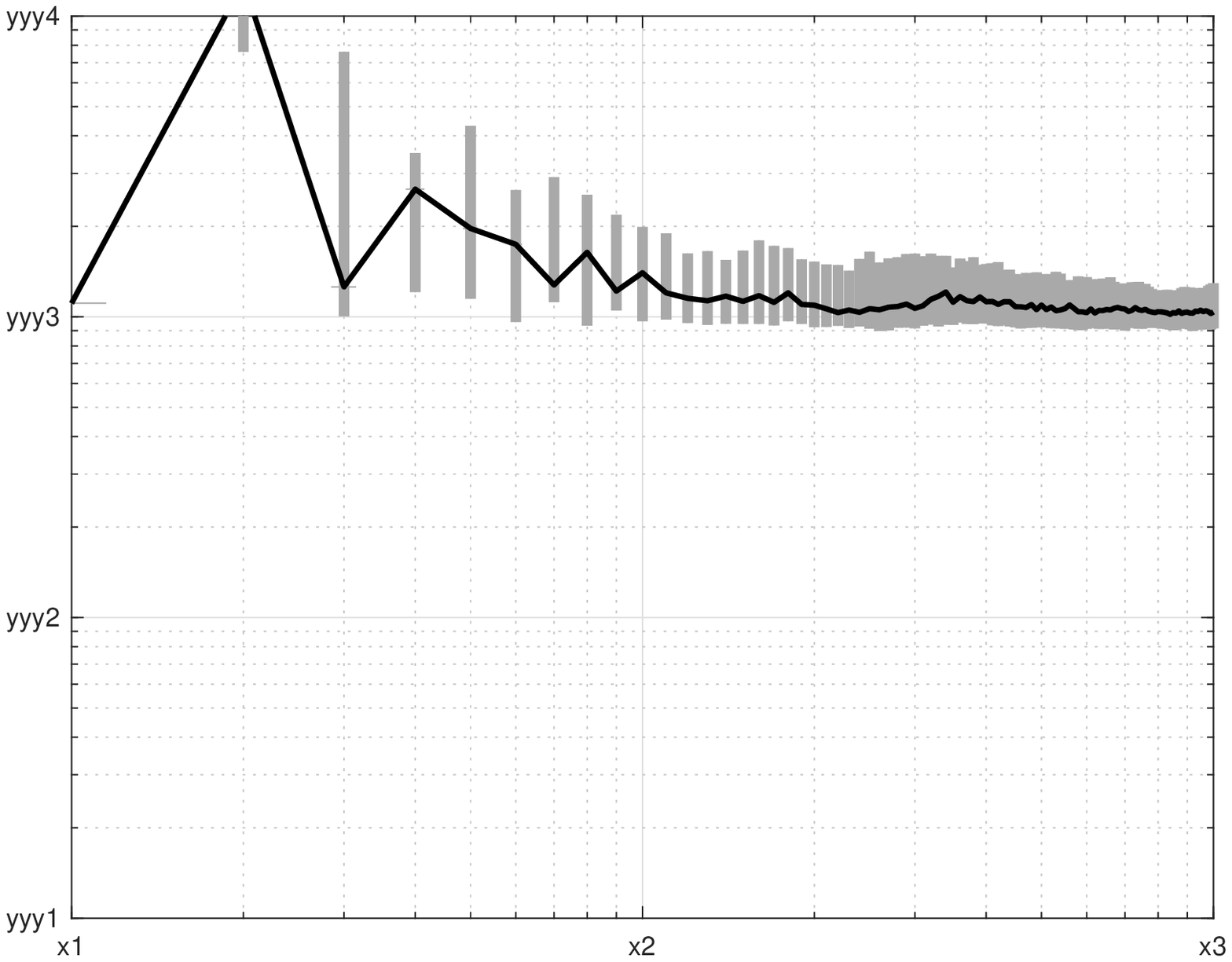}};
\node[] at (0,-2.3) {\footnotesize $k$};
\node[rotate=90] at (-2.9,0) {\footnotesize $\psi(\bar{\theta}[k])$};
\end{tikzpicture}
\hspace{-.3in}
&
\hspace{-.2in}
\begin{tikzpicture}
\node[] at (0,0) {
\psfrag{x1}[cc][][0.6][0]{\quad\quad$10^0$}
\psfrag{x2}[cc][][0.6][0]{$10^{1}$}
\psfrag{x3}[cc][][0.6][0]{$10^{2}$}
\psfrag{yyy1}[cl][bl][0.6][45]{$10^{-1}$}
\psfrag{yyy2}[cc][][0.6][45]{$10^{0}$}
\psfrag{yyy3}[cc][][0.6][45]{$10^{1}$}
\psfrag{yyy4}[cc][][0.6][45]{$10^{2}$}
\psfrag{yyy5}[cc][][0.6][45]{$10^{3}$}
\includegraphics[width=.35\linewidth]{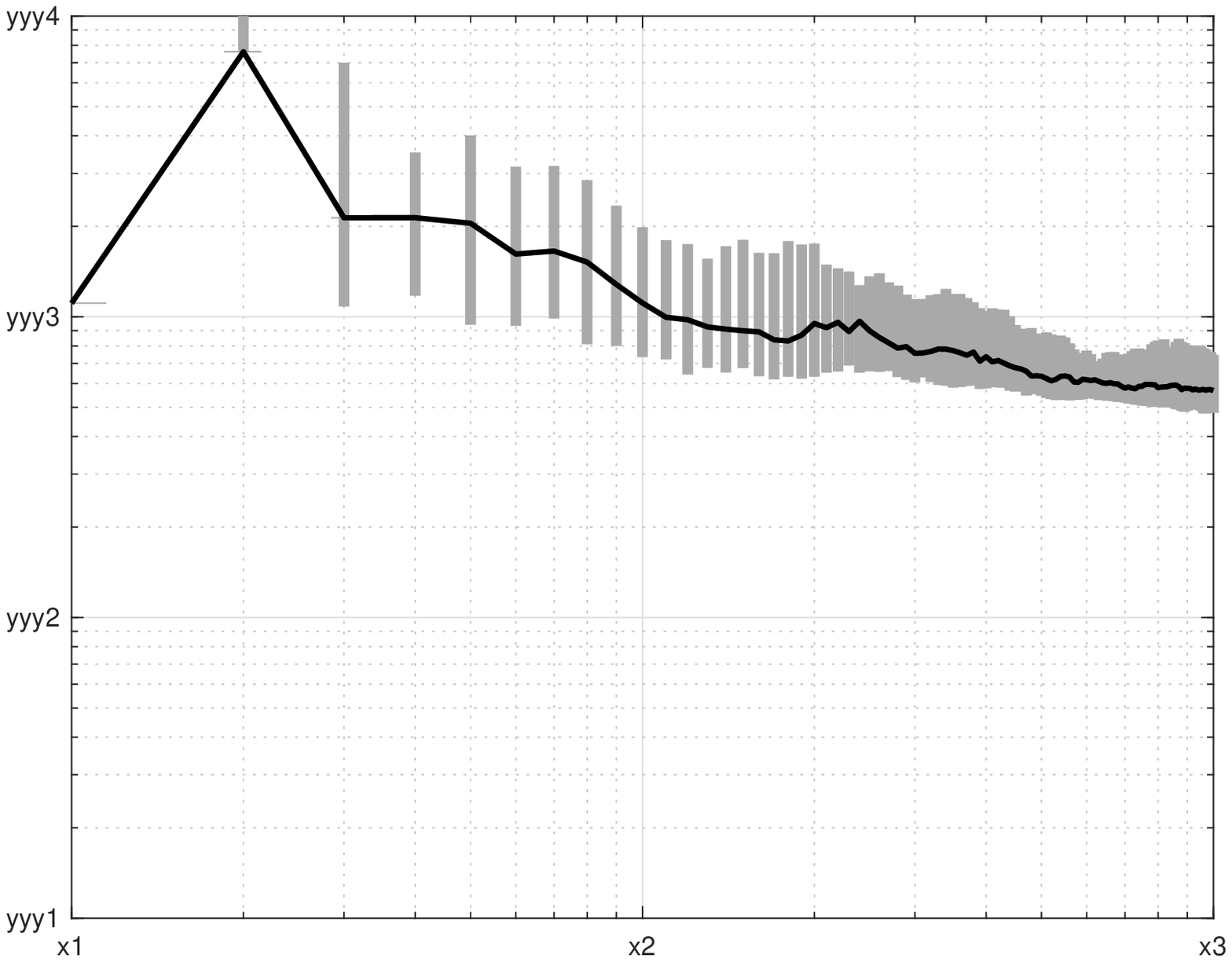}};
\node[] at (0,-2.3) {\footnotesize $k$};
\node[rotate=90] at (-2.9,0) {\footnotesize $\psi(\bar{\theta}[k])$};
\end{tikzpicture}
\hspace{-.3in}
&
\hspace{-.2in}
\begin{tikzpicture}
\node[] at (0,0) {
\psfrag{x1}[cc][][0.6][0]{\quad\quad$10^0$}
\psfrag{x2}[cc][][0.6][0]{$10^{1}$}
\psfrag{x3}[cc][][0.6][0]{$10^{2}$}
\psfrag{yyy1}[cl][bl][0.6][45]{$10^{-1}$}
\psfrag{yyy2}[cc][][0.6][45]{$10^{0}$}
\psfrag{yyy3}[cc][][0.6][45]{$10^{1}$}
\psfrag{yyy4}[cc][][0.6][45]{$10^{2}$}
\psfrag{yyy5}[cc][][0.6][45]{$10^{3}$}
\includegraphics[width=.35\linewidth]{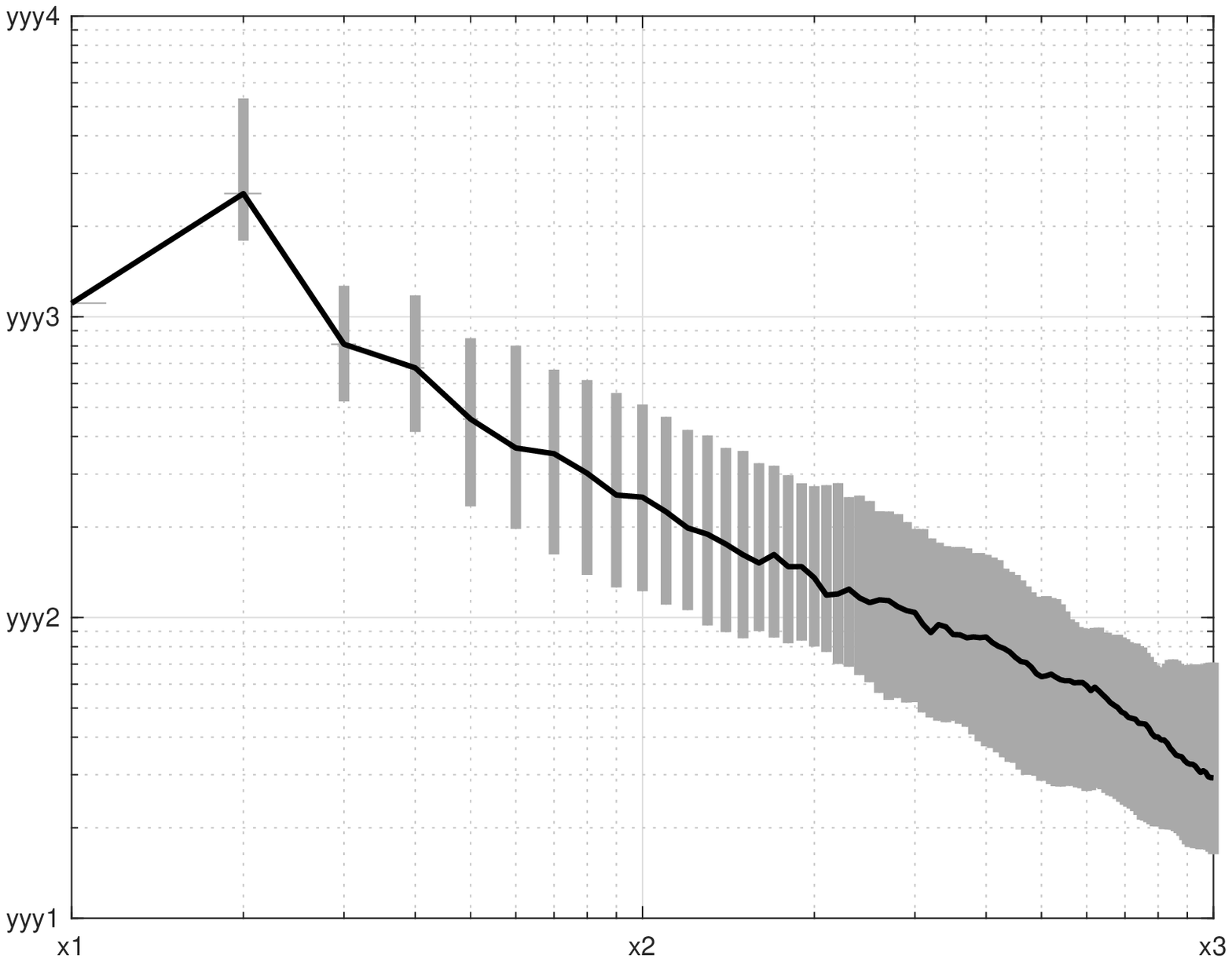}};
\node[] at (0,-2.3) {\footnotesize $k$};
\node[rotate=90] at (-2.9,0) {\footnotesize $\psi(\bar{\theta}[k])$};
\end{tikzpicture}
\end{tabular}
\vspace{-.2in}
\caption{
\label{fig:a0}  
Statistics of relative fitness of the stochastic gradient method in Algorithm~\ref{alg:1} for learning lending interest rates versus the iteration number for $T=100$ with various choices of privacy budgets. The boxes, i.e., the vertical lines at each iterations, illustrate the range of 25\% to 75\% percentiles for extracted from a hundred runs of the algorithm and the black lines show the median relative fitness. }
\end{figure*}

\begin{figure}[t]
\centering
\begin{tikzpicture}
\node[] at (0,0) {
\psfrag{x1}[cc][][0.7][0]{$10^{3}$}
\psfrag{x2}[cc][][0.7][0]{$10^{4}$}
\psfrag{x3}[cc][][0.7][0]{$10^{5}$}
\psfrag{y1}[cc][][0.7][0]{$10^{-1}$\quad}
\psfrag{y2}[cc][][0.7][0]{$10^{0}$\quad}
\psfrag{y3}[cc][][0.7][0]{$10^{1}$\quad}
\psfrag{z1}[cc][][0.7][0]{$10^{0}$\quad}
\psfrag{z2}[cc][][0.7][0]{$10^{5}$\quad}
\psfrag{z3}[cc][][0.7][0]{$10^{10}$\quad}
\includegraphics[width=.75\columnwidth]{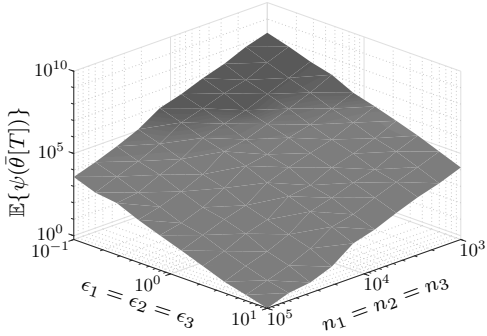}};
\node[rotate=20] at (+1.7,-1.9) {\footnotesize $n_1=n_2=n_3$};
\node[rotate=-18] at (-1.6,-1.9) {\footnotesize $\epsilon_1=\epsilon_2=\epsilon_3$};
\node[rotate=90] at (-3.2,0) {\footnotesize $\mathbb{E}\{\psi(\bar{\theta}[T])\}$};
\end{tikzpicture}
\vspace{-.1in}
\caption{
\label{fig:a1} Relative fitness of the stochastic gradient method in Algorithm~\ref{alg:1} for learning lending interest rates after $T=100$ iterations versus the size of the datasets and the privacy budgets. }
\end{figure}

\section{Experimental Validation of the Performance of ML on Distributed Private Data} \label{sec:numerical}
In this section, we examine the results of the paper, specifically the performance of Algorithm~\ref{alg:1}, on two financial datasets on lending and credit card fraud. Particularly, we use the relative fitness of the iterates in Algorithm~\ref{alg:1} to illustrate its performance. The relative fitness of $\theta$ is given by
\begin{align}
    \psi(\theta):=\frac{f(\theta)}{f(\theta^*)}-1.
\end{align}
This measure shows how good $\theta$ is in comparison to the optimal ML model $\theta^*$ in terms of the training cost in~\eqref{eqn:ML}. We opt for studying the relative fitness, scaled by $f(\theta^*)$ as opposed as the absolute fitness $f(\theta)-f(\theta^*)$, because we consider datasets with different sizes for two distinct ML learning models and thus we want to factor out the effects of the variations of $f(\theta^*)$. Finally, note that, by construction, $\psi(\theta)\geq 0$. Further, the lower the value of $\psi(\theta)$, the better $\theta$ performs in comparison to $\theta^*$. \farhad{In what follows, we use Algorithm~\ref{alg:0} with $\epsilon=+\infty$ for non-private learning of $\theta^*$; this is equivalent to setting the magnitude of the additive privacy-preserving noise in the gradients to zero.}

\subsection{Lending Dataset}
First, we use a lending dataset with a linear regression model to demonstrate the value of the methodology and to validate the theoretical results. 

\begin{figure}[t]
\centering
\begin{tikzpicture}
\node[] at (0,0) {
\psfrag{x1}[cc][][0.6][0]{$10^{-1}$}
\psfrag{x2}[cc][][0.6][0]{$10^{0}$}
\psfrag{x3}[cc][][0.6][0]{$10^{1}$}
\psfrag{y1}[cc][][0.6][0]{$10^{-1}$\quad}
\psfrag{y2}[cc][][0.6][0]{$10^{1}$\quad}
\psfrag{y3}[cc][][0.6][0]{$10^{3}$\quad}
\psfrag{y4}[cc][][0.6][0]{$10^{5}$\quad}
\psfrag{y5}[cc][][0.6][0]{$10^{7}$\quad}
\psfrag{y6}[cc][][0.6][0]{$10^{9}$\quad}
\psfrag{aaaaaaaaaa1}[l][t][0.6][0]{\small \hspace{-.299in}\raisebox{-10pt}{$n=10^3$}}
\psfrag{aaaaaaaaaa2}[l][t][0.6][0]{\small \hspace{-.31in}\raisebox{-10pt}{$n=10^4$}}
\psfrag{aaaaaaaaaa3}[l][t][0.6][0]{\small \hspace{-.3in}\raisebox{-10pt}{$n=10^5$}}
\includegraphics[width=.75\columnwidth]{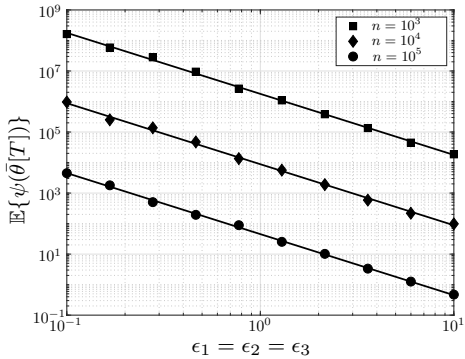}};
\node[] at (0,-2.4) {\footnotesize $\epsilon_1=\epsilon_2=\epsilon_3$};
\node[rotate=90] at (-3.1,0) {\footnotesize $\mathbb{E}\{\psi(\bar{\theta}[T])\}$};
\end{tikzpicture}
\vspace{-.1in}
\caption{
\label{fig:a3} Relative fitness of the stochastic gradient method in Algorithm~\ref{alg:1} for learning lending interest rates after $T=100$ iterations versus the privacy budgets. The solid line illustrate the bound in Theorem~\ref{tho:utility_strongconvex}.  }
\end{figure}

\begin{figure}
\centering
\begin{tikzpicture}
\node[] at (0,0) {
\psfrag{x1}[cc][][0.6][0]{$10^{-1}$}
\psfrag{x2}[cc][][0.6][0]{$10^{0}$}
\psfrag{x3}[cc][][0.6][0]{$10^{1}$}
\psfrag{y1}[cc][][0.6][0]{$10^{-1}$\quad}
\psfrag{y2}[cc][][0.6][0]{$10^{1}$\quad}
\psfrag{y3}[cc][][0.6][0]{$10^{3}$\quad}
\psfrag{y4}[cc][][0.6][0]{$10^{5}$\quad}
\psfrag{y5}[cc][][0.6][0]{$10^{7}$\quad}
\psfrag{y6}[cc][][0.6][0]{$10^{9}$\quad}
\psfrag{aaaaaaa1}[l][t][0.6][0]{\small \hspace{-.24in}\raisebox{-10pt}{$\epsilon=0.1$}}
\psfrag{aaaaaaa2}[l][t][0.6][0]{\small \hspace{-.24in}\raisebox{-10pt}{$\epsilon=1$}}
\psfrag{aaaaaaa3}[l][t][0.6][0]{\small \hspace{-.24in}\raisebox{-10pt}{$\epsilon=10$}}
\includegraphics[width=.75\columnwidth]{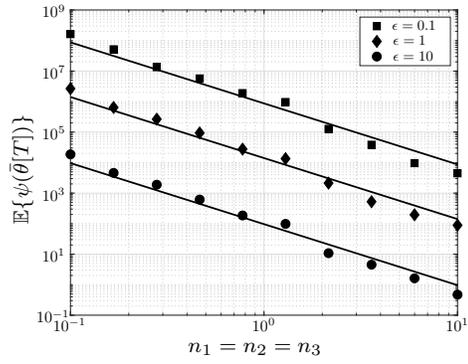}};
\node[] at (0,-2.4) {\footnotesize $n_1=n_2=n_3$};
\node[rotate=90] at (-3.1,0) {\footnotesize $\mathbb{E}\{\psi(\bar{\theta}[T])\}$};
\end{tikzpicture}
\vspace{-.1in}
\caption{\label{fig:a4} Relative fitness of the stochastic gradient method in Algorithm~\ref{alg:1} for learning lending interest rates after $T=100$ iterations versus the size of the datasets. The solid line illustrate the bound in Theorem~\ref{tho:utility_strongconvex}.  }
\end{figure}

\begin{figure*}
\centering
\begin{tabular}{cccc}
\hspace{-.2in}
\footnotesize 
Scenario 1:
\hspace{-.2in}
&
\hspace{-.2in}
\footnotesize 
Scenario 2:
\hspace{-.2in}
&
\hspace{-.2in}
\footnotesize 
Scenario 3:
\hspace{-.2in}
&
\hspace{-.2in}
\footnotesize 
Scenario 4:\\[-.5em]
\hspace{-.2in}
\footnotesize 
$n_2=n_3=10^3$
\hspace{-.2in}
&
\hspace{-.2in}
\footnotesize 
$n_2=n_3=10^3$
\hspace{-.2in}
&
\hspace{-.2in}
\footnotesize 
$n_2=n_3=10^5$
\hspace{-.2in}
&
\hspace{-.2in}
\footnotesize 
$n_2=n_3=10^5$\\[-.4em]
\hspace{-.2in}
\footnotesize 
(small dataset)
\hspace{-.2in}
&
\hspace{-.2in}
\footnotesize 
(small dataset)
\hspace{-.2in}
&
\hspace{-.2in}
\footnotesize 
(large dataset)
\hspace{-.2in}
&
\hspace{-.2in}
\footnotesize 
(large dataset)\\[-.5em]
\hspace{-.2in}
\footnotesize 
$\epsilon_2=\epsilon_3=0.1$ 
\hspace{-.2in}
&
\hspace{-.2in}
\footnotesize 
$\epsilon_2=\epsilon_3=10$ 
\hspace{-.2in}
&
\hspace{-.2in}
\footnotesize 
$\epsilon_2=\epsilon_3=0.1$ 
\hspace{-.2in}
&
\hspace{-.2in}
\footnotesize 
$\epsilon_2=\epsilon_3=10$ \\[-.5em]
\hspace{-.2in}
\footnotesize 
(small privacy budget)
\hspace{-.2in}
&
\hspace{-.2in}
\footnotesize 
(large privacy budget)
\hspace{-.2in}
&
\hspace{-.2in}
\footnotesize 
(small privacy budget)
\hspace{-.2in}
&
\hspace{-.2in}
\footnotesize 
(large privacy budget) \\[-.5em]
\hspace{-.2in}
\begin{tikzpicture}
\node[] at (0,0) {
\psfrag{x1}[cc][][0.6][0]{$10^3$}
\psfrag{x2}[cc][][0.6][0]{$10^4$}
\psfrag{x3}[cc][][0.6][0]{\quad$10^5$}
\psfrag{y1}[cc][bl][0.6][0]{\;$10^{-1}$}
\psfrag{y2}[cc][][0.6][0]{$10^{0}$}
\psfrag{y3}[cc][][0.6][0]{$10^{1}$\quad}
\psfrag{z1}[cc][t][0.6][0]{$10^{4}$\quad}
\psfrag{z2}[cc][][0.6][0]{$10^{5}$\quad}
\psfrag{z3}[cc][][0.6][0]{$10^{6}$\quad}
\psfrag{z4}[cc][][0.6][0]{$10^{7}$\quad}
\includegraphics[width=.25\linewidth]{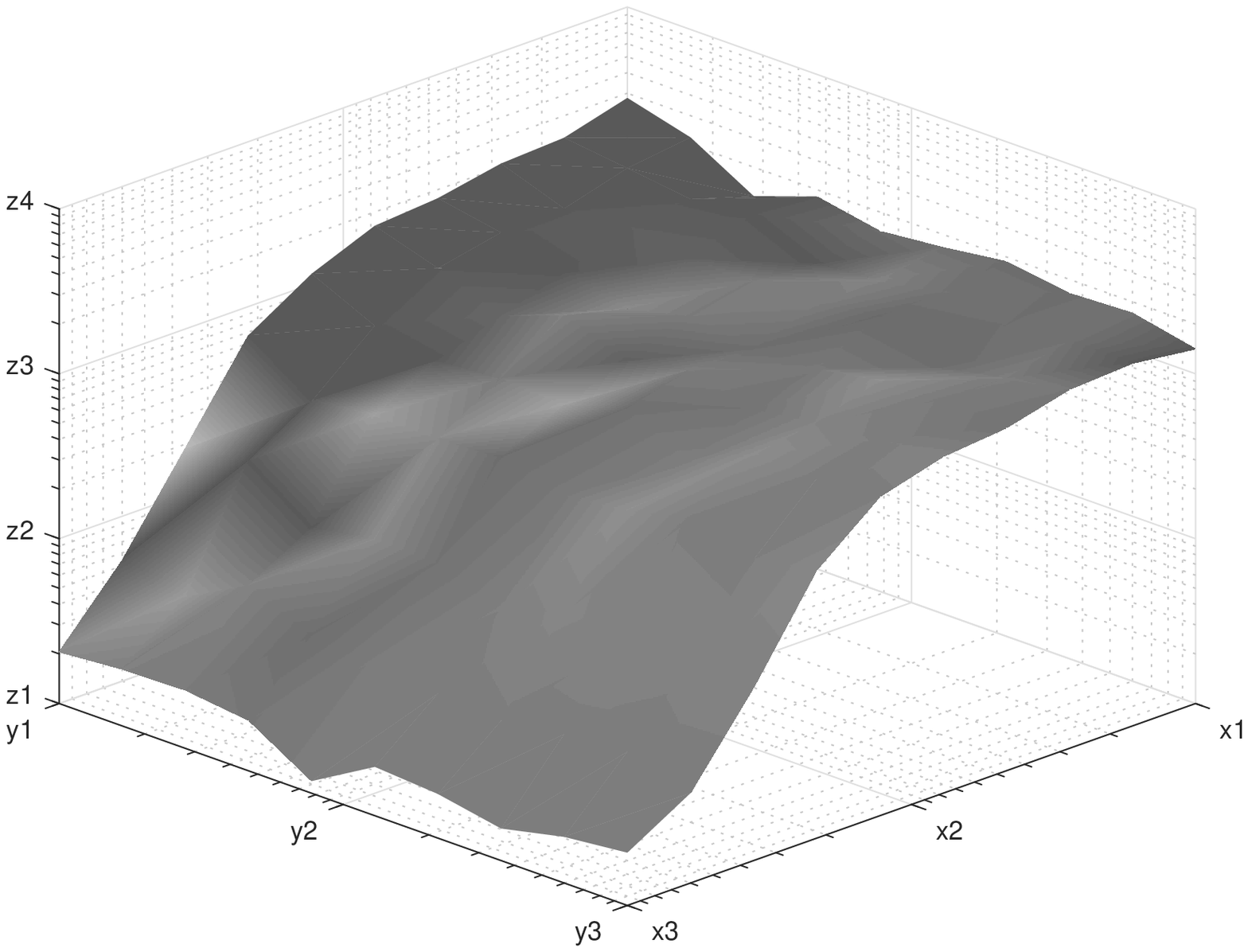}};
\node[] at (-1.1,-1.4) {\footnotesize $\epsilon_1$};
\node[] at (+1.2,-1.4) {\footnotesize $n_1$};
\node[rotate=90] at (-2.2,0) {\footnotesize $\mathbb{E}\{\psi(\bar{\theta}[T])\}$};
\end{tikzpicture}
\hspace{-.2in}
&
\hspace{-.2in}
\begin{tikzpicture}
\node[] at (0,0) {
\psfrag{x1}[cc][][0.6][0]{$10^3$}
\psfrag{x2}[cc][][0.6][0]{$10^4$}
\psfrag{x3}[cc][][0.6][0]{\quad$10^5$}
\psfrag{y1}[cc][bl][0.6][0]{\;$10^{-1}$}
\psfrag{y2}[cc][][0.6][0]{$10^{0}$}
\psfrag{y3}[cc][][0.6][0]{$10^{1}$\quad}
\psfrag{z1}[cc][t][0.6][0]{$10^{0}$\quad}
\psfrag{z2}[cc][][0.6][0]{$10^{2}$\quad}
\psfrag{z3}[cc][][0.6][0]{$10^{4}$\quad}
\psfrag{z4}[cc][][0.6][0]{$10^{6}$\quad}
\psfrag{z5}[cc][][0.6][0]{$10^{8}$\quad}
\includegraphics[width=.25\linewidth]{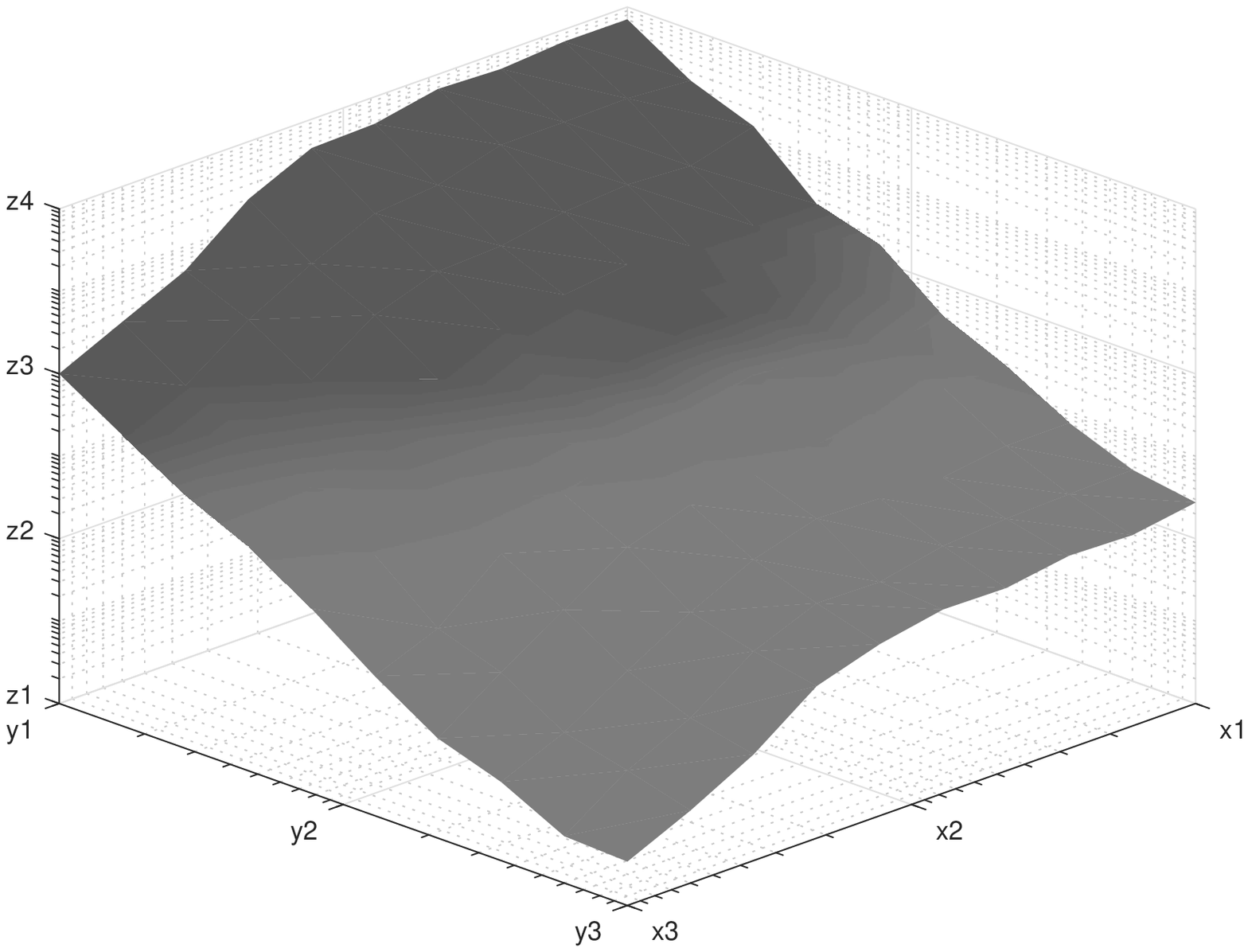}};
\node[] at (-1.1,-1.4) {\footnotesize $\epsilon_1$};
\node[] at (+1.2,-1.4) {\footnotesize $n_1$};
\node[rotate=90] at (-2.2,0) {\footnotesize $\mathbb{E}\{\psi(\bar{\theta}[T])\}$};
\end{tikzpicture}
\hspace{-.2in}
&
\hspace{-.2in}
\begin{tikzpicture}
\node[] at (0,0) {
\psfrag{x1}[cc][][0.6][0]{$10^3$}
\psfrag{x2}[cc][][0.6][0]{$10^4$}
\psfrag{x3}[cc][][0.6][0]{\quad$10^5$}
\psfrag{y1}[cc][bl][0.6][0]{\;$10^{-1}$}
\psfrag{y2}[cc][][0.6][0]{$10^{0}$}
\psfrag{y3}[cc][][0.6][0]{$10^{1}$\quad}
\psfrag{z1}[cc][t][0.6][0]{$10^{3}$\quad}
\psfrag{z2}[cc][][0.6][0]{$10^{4}$\quad}
\includegraphics[width=.25\linewidth]{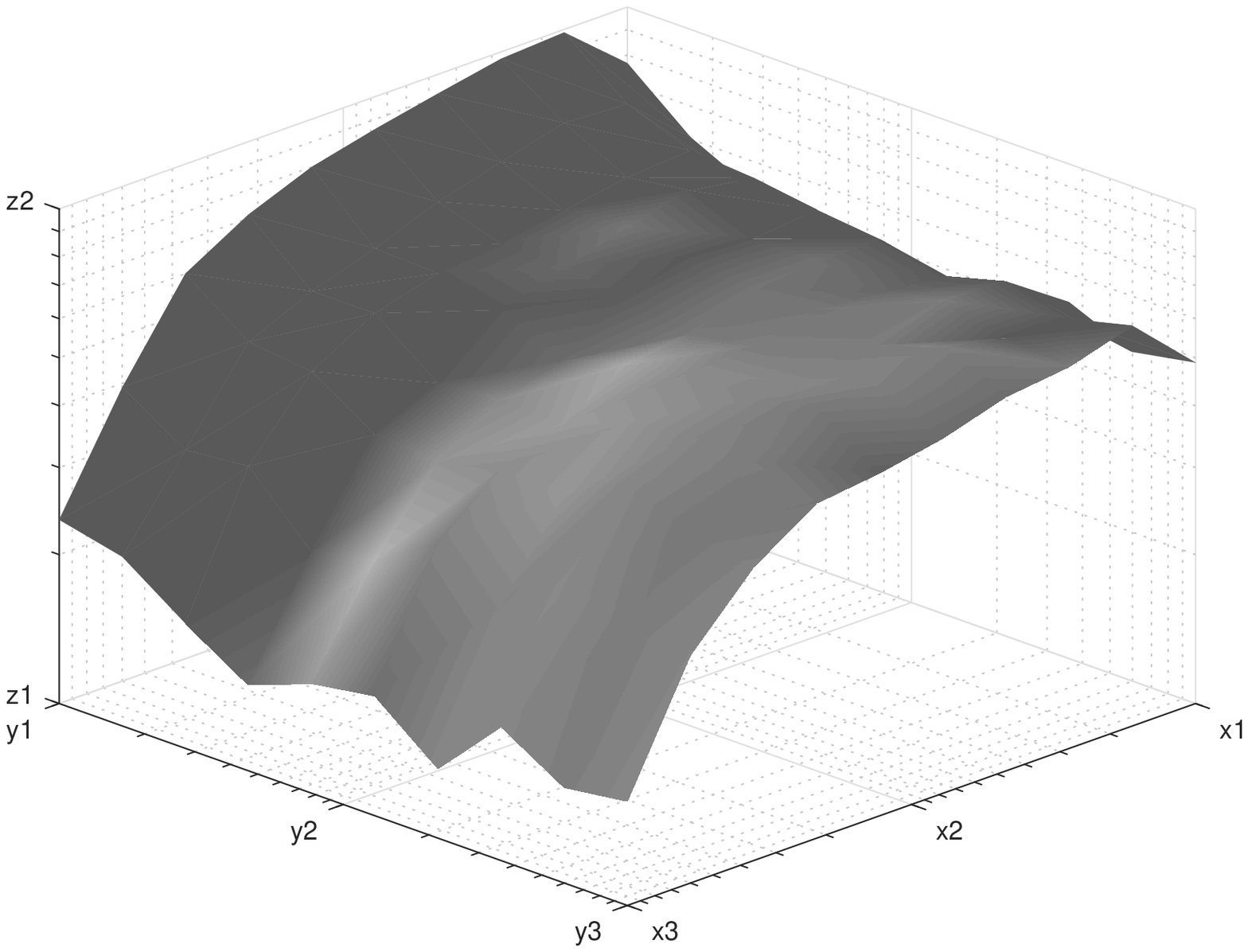}};
\node[] at (-1.1,-1.4) {\footnotesize $\epsilon_1$};
\node[] at (+1.2,-1.4) {\footnotesize $n_1$};
\node[rotate=90] at (-2.2,0) {\footnotesize $\mathbb{E}\{\psi(\bar{\theta}[T])\}$};
\end{tikzpicture}
\hspace{-.2in}
&
\hspace{-.2in}
\begin{tikzpicture}
\node[] at (0,0) {
\psfrag{x1}[cc][][0.6][0]{$10^3$}
\psfrag{x2}[cc][][0.6][0]{$10^4$}
\psfrag{x3}[cc][][0.6][0]{\quad$10^5$}
\psfrag{y1}[cc][bl][0.6][0]{\;$10^{-1}$}
\psfrag{y2}[cc][][0.6][0]{$10^{0}$}
\psfrag{y3}[cc][][0.6][0]{$10^{1}$\quad}
\psfrag{z1}[cc][t][0.6][0]{$10^{0}$\quad}
\psfrag{z2}[cc][][0.6][0]{$10^{1}$\quad}
\psfrag{z3}[cc][][0.6][0]{$10^{2}$\quad}
\psfrag{z4}[cc][][0.6][0]{$10^{3}$\quad}
\psfrag{z5}[cc][][0.6][0]{$10^{4}$\quad}
\includegraphics[width=.25\linewidth]{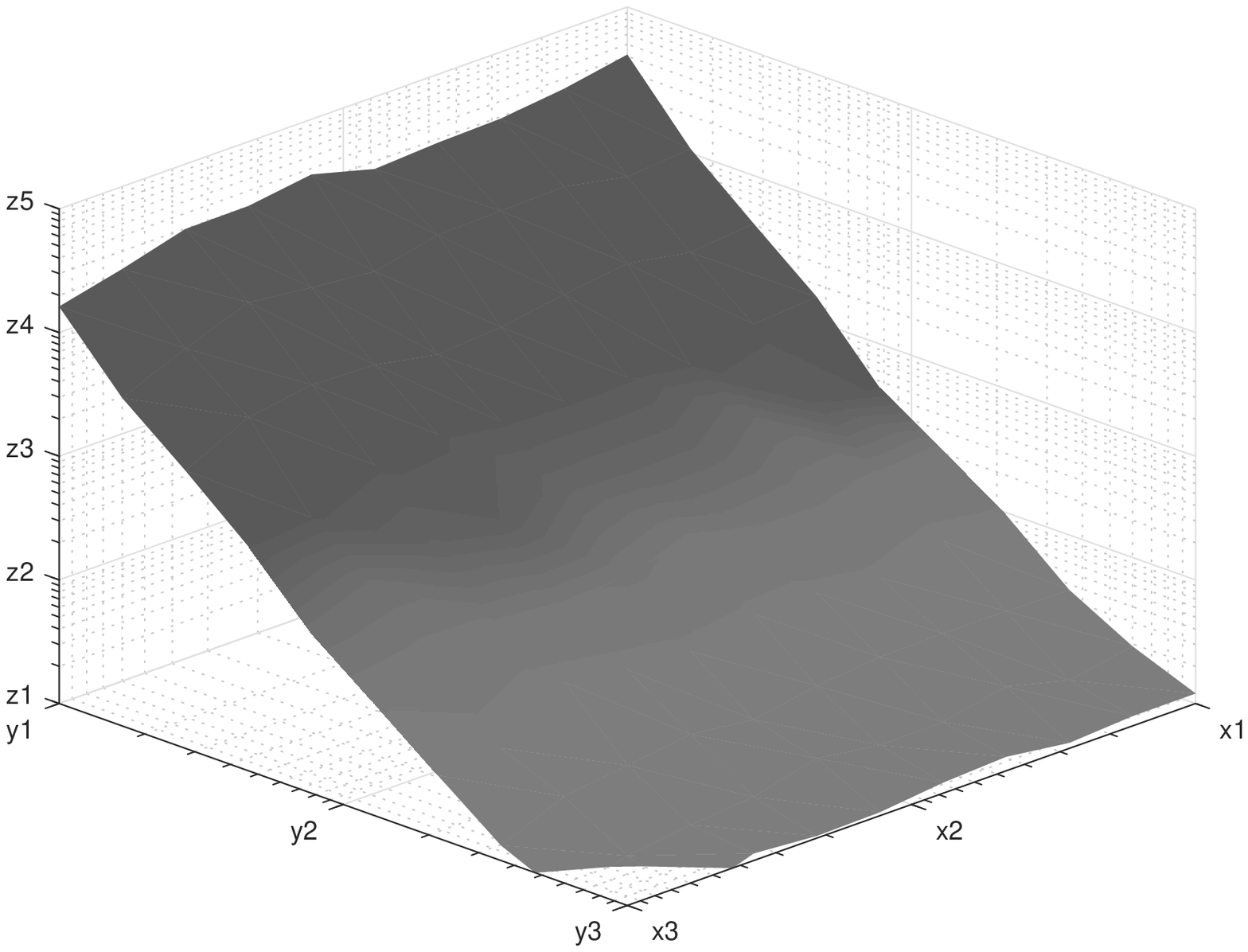}};
\node[] at (-1.1,-1.4) {\footnotesize $\epsilon_1$};
\node[] at (+1.2,-1.4) {\footnotesize $n_1$};
\node[rotate=90] at (-2.2,0) {\footnotesize $\mathbb{E}\{\psi(\bar{\theta}[T])\}$};
\end{tikzpicture}
\end{tabular}
\vspace{-.1in}
\caption{
\label{fig:a5}  
Relative fitness of the stochastic gradient method in Algorithm~\ref{alg:1} for learning lending interest rates after $T=100$ iterations versus the size of the dataset and the  privacy budget of the first data owner for four distinct scenarios of collaboration.
 }
\end{figure*}

\begin{figure*}
\centering
\begin{tabular}{ccc}
\hspace{-.25in}
$\epsilon_1=\epsilon_2=\epsilon_3=0.1$
\hspace{-.3in}
&
\hspace{-.2in}
$\epsilon_1=\epsilon_2=\epsilon_3=1.0$
\hspace{-.3in}
&
\hspace{-.2in}
$\epsilon_1=\epsilon_2=\epsilon_3=10$
\hspace{-.2in}
\\[-.5em]
\hspace{-.25in}
\begin{tikzpicture}
\node[] at (0,0) {
\psfrag{x1}[cc][][0.6][0]{\quad\quad$10^0$}
\psfrag{x2}[cc][][0.6][0]{$10^{1}$}
\psfrag{x3}[cc][][0.6][0]{$10^{2}$}
\psfrag{yyy1}[cl][bl][0.6][45]{\hspace{-.07in}$10^{-2}$}
\psfrag{yyy2}[cc][][0.6][45]{$10^{-1}$}
\psfrag{yyy3}[cc][][0.6][45]{$10^{0}$}
\psfrag{yyy4}[cc][][0.6][45]{$10^{1}$}
\psfrag{yyy5}[cc][][0.6][45]{$10^{2}$}
\psfrag{yyy6}[cc][][0.6][45]{$10^{3}$}
\includegraphics[width=.35\linewidth]{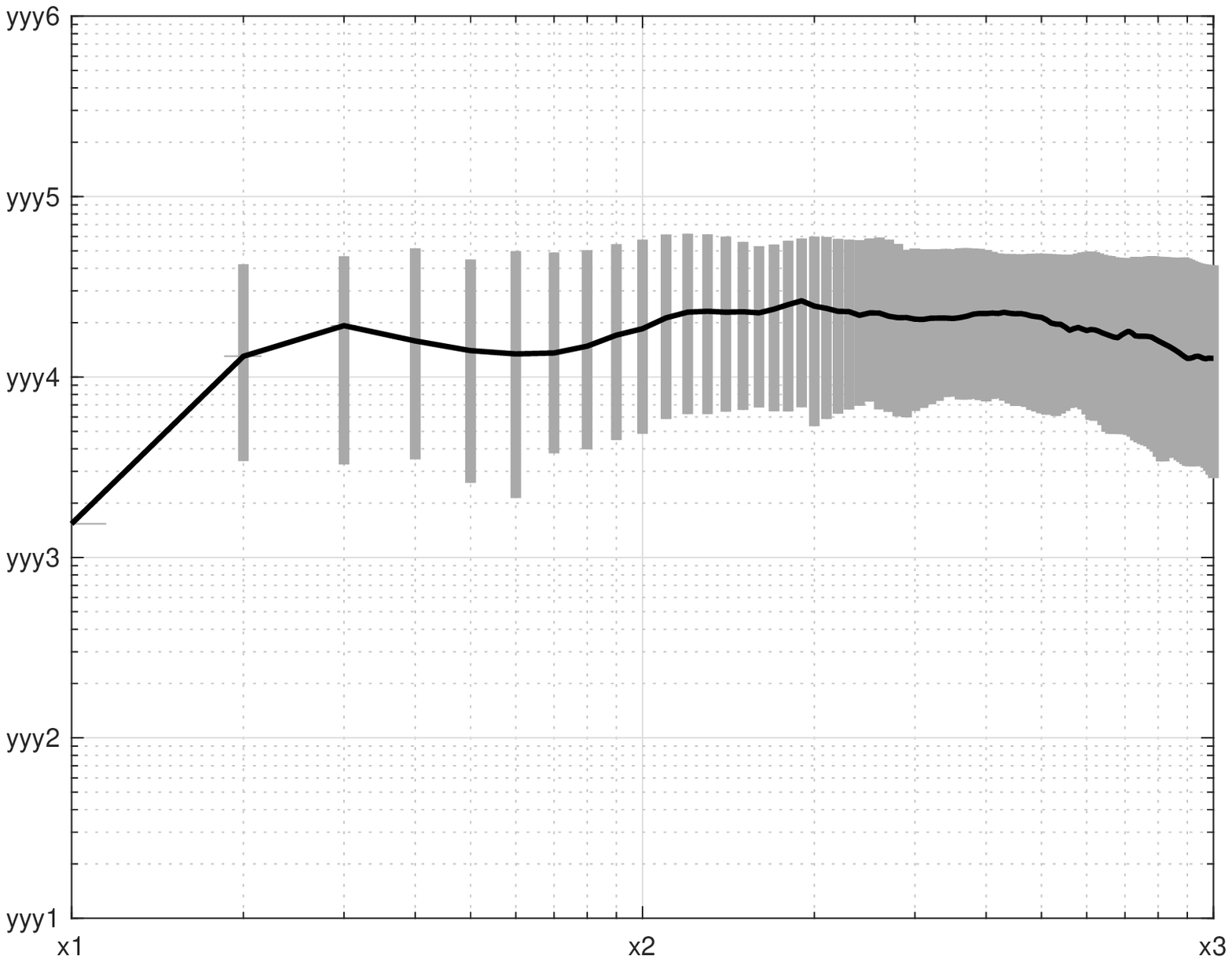}};
\node[] at (0,-2.3) {\footnotesize $k$};
\node[rotate=90] at (-2.9,0) {\footnotesize $\psi(\bar{\theta}[k])$};
\end{tikzpicture}
\hspace{-.3in}
&
\hspace{-.2in}
\begin{tikzpicture}
\node[] at (0,0) {
\psfrag{x1}[cc][][0.6][0]{\quad\quad$10^0$}
\psfrag{x2}[cc][][0.6][0]{$10^{1}$}
\psfrag{x3}[cc][][0.6][0]{$10^{2}$}
\psfrag{yyy1}[cl][bl][0.6][45]{\hspace{-.07in}$10^{-2}$}
\psfrag{yyy2}[cc][][0.6][45]{$10^{-1}$}
\psfrag{yyy3}[cc][][0.6][45]{$10^{0}$}
\psfrag{yyy4}[cc][][0.6][45]{$10^{1}$}
\psfrag{yyy5}[cc][][0.6][45]{$10^{2}$}
\psfrag{yyy6}[cc][][0.6][45]{$10^{3}$}
\includegraphics[width=.35\linewidth]{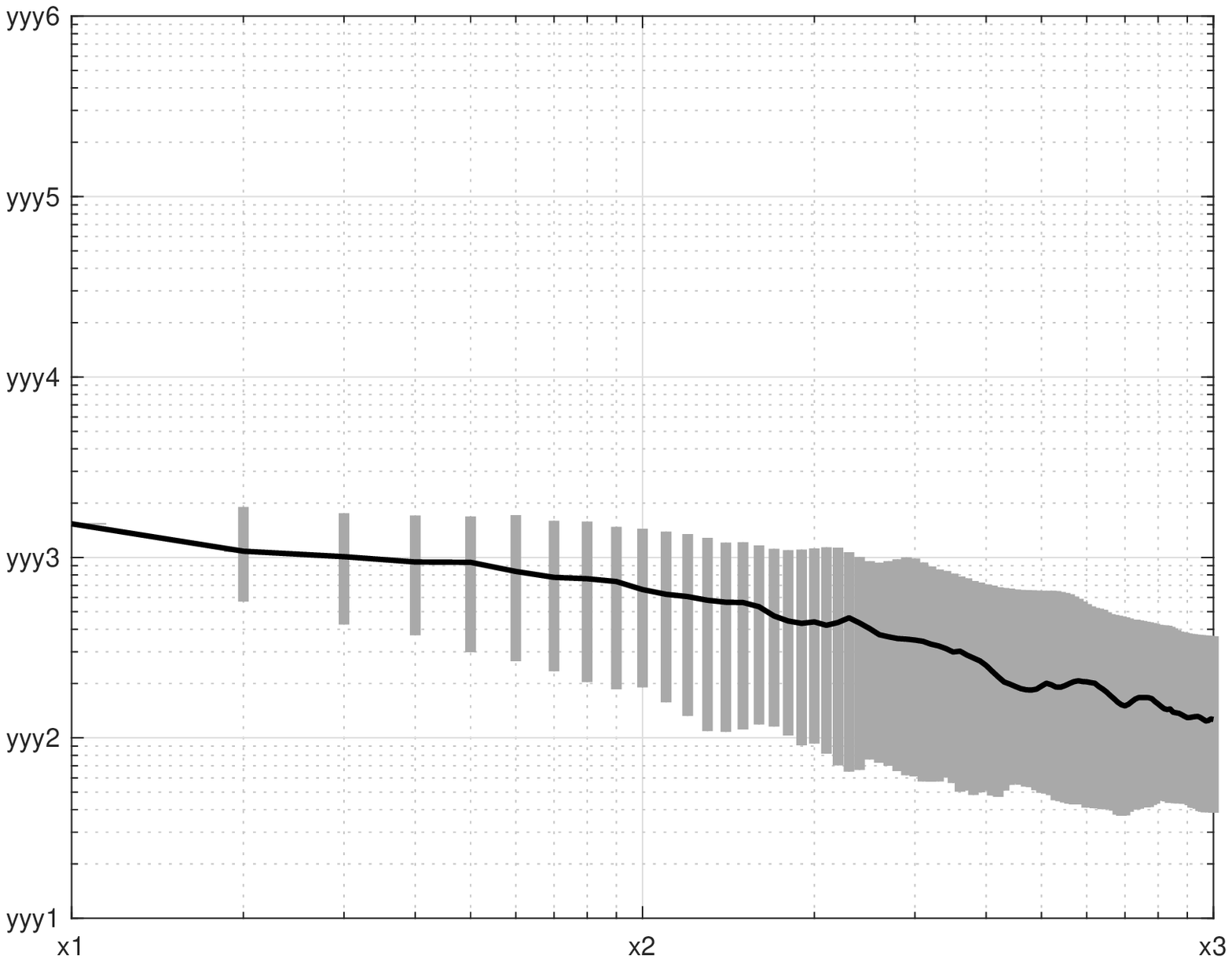}};
\node[] at (0,-2.3) {\footnotesize $k$};
\node[rotate=90] at (-2.9,0) {\footnotesize $\psi(\bar{\theta}[k])$};
\end{tikzpicture}
\hspace{-.3in}
&
\hspace{-.2in}
\begin{tikzpicture}
\node[] at (0,0) {
\psfrag{x1}[cc][][0.6][0]{\quad\quad$10^0$}
\psfrag{x2}[cc][][0.6][0]{$10^{1}$}
\psfrag{x3}[cc][][0.6][0]{$10^{2}$}
\psfrag{yyy1}[cl][bl][0.6][45]{\hspace{-.07in}$10^{-2}$}
\psfrag{yyy2}[cc][][0.6][45]{$10^{-1}$}
\psfrag{yyy3}[cc][][0.6][45]{$10^{0}$}
\psfrag{yyy4}[cc][][0.6][45]{$10^{1}$}
\psfrag{yyy5}[cc][][0.6][45]{$10^{2}$}
\psfrag{yyy6}[cc][][0.6][45]{$10^{3}$}
\includegraphics[width=.35\linewidth]{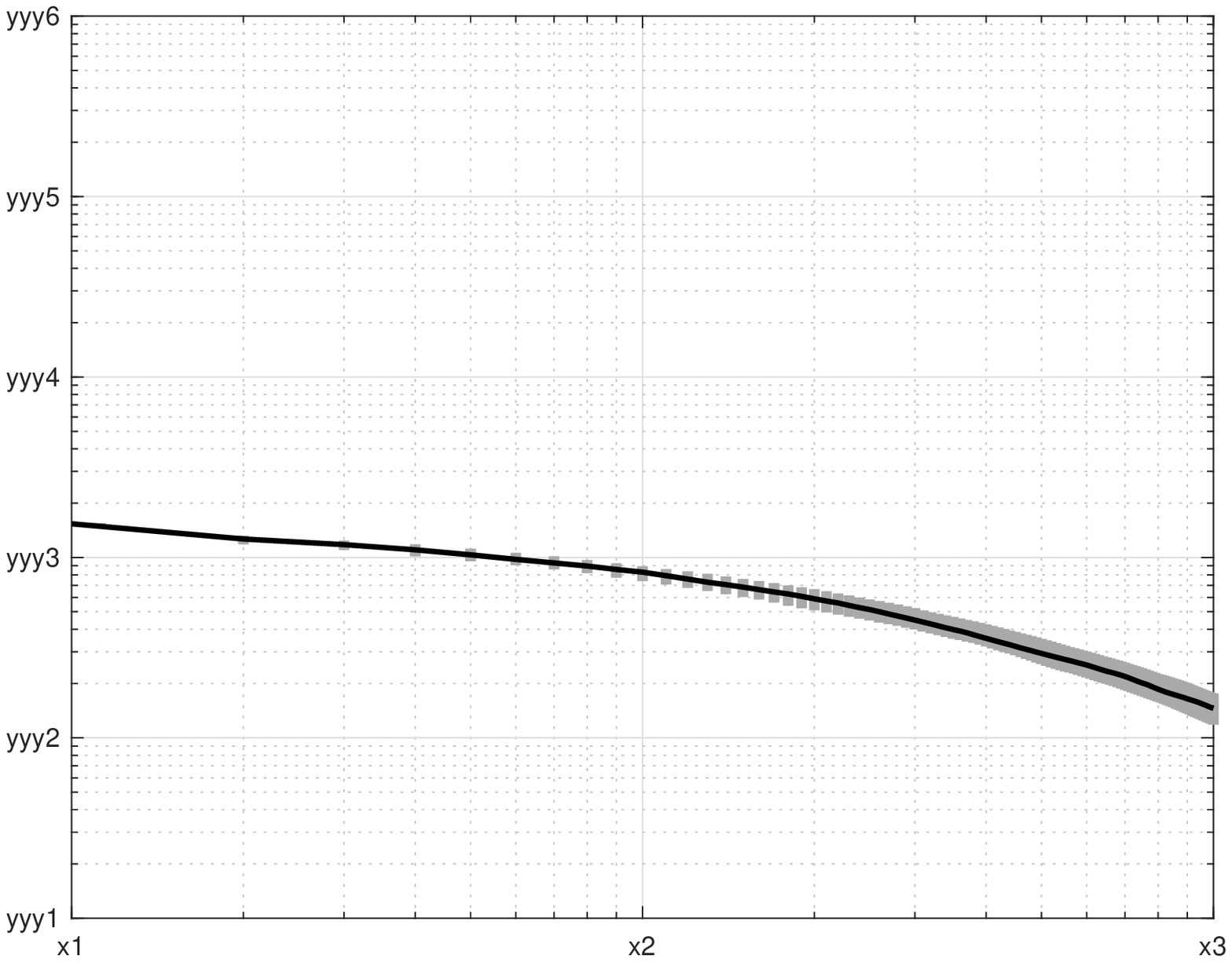}};
\node[] at (0,-2.3) {\footnotesize $k$};
\node[rotate=90] at (-2.9,0) {\footnotesize $\psi(\bar{\theta}[k])$};
\end{tikzpicture}
\end{tabular}
\vspace{-.2in}
\caption{
\label{fig:b0}  
Statistics of relative fitness of the stochastic gradient method in Algorithm~\ref{alg:1} for fraud detection versus the iteration number for $T=100$ with various choices of  privacy budgets. The boxes, i.e., the vertical lines at each iterations, illustrate the range of 25\% to 75\% percentiles for extracted from a hundred runs of the algorithm and the black lines show the median relative fitness.  }
\end{figure*}

\begin{figure}[t]
\centering
\begin{tikzpicture}
\node[] at (0,0) {
\psfrag{x1}[cc][][0.7][0]{$10^{3}$}
\psfrag{x2}[cc][][0.7][0]{$10^{4}$}
\psfrag{x3}[cc][][0.7][0]{$10^{5}$}
\psfrag{y1}[cc][][0.7][0]{\raisebox{-20pt}{$10^{-1}$}}
\psfrag{y2}[cc][][0.7][0]{$10^{0}$\quad}
\psfrag{y3}[cc][][0.7][0]{$10^{1}$\quad}
\psfrag{z1}[cc][][0.7][0]{$10^{-2}$\quad}
\psfrag{z2}[cc][][0.7][0]{$10^{0}$\quad}
\psfrag{z3}[cc][][0.7][0]{$10^{2}$\quad}
\psfrag{z4}[cc][][0.7][0]{$10^{4}$\quad}
\psfrag{z5}[cc][][0.7][0]{$10^{6}$\quad}
\psfrag{z6}[cc][][0.7][0]{$10^{8}$\quad}
\includegraphics[width=.75\columnwidth]{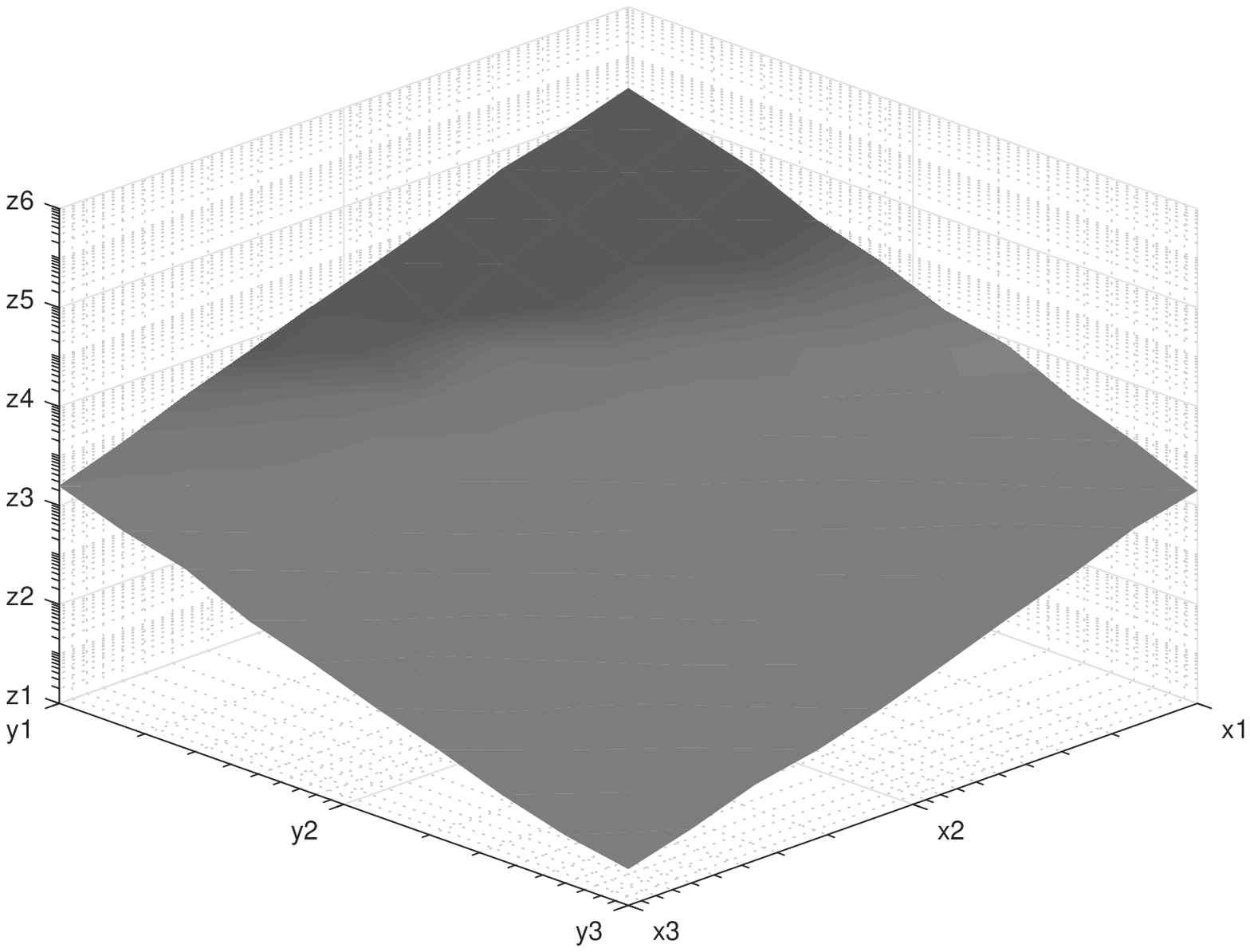}};
\node[rotate=20] at (+1.7,-1.9) {\footnotesize $n_1=n_2=n_3$};
\node[rotate=-18] at (-1.6,-1.9) {\footnotesize $\epsilon_1=\epsilon_2=\epsilon_3$};
\node[rotate=90] at (-3.2,0) {\footnotesize $\mathbb{E}\{\psi(\bar{\theta}[T])\}$};
\end{tikzpicture}
\vspace{-.1in}
\caption{\label{fig:b1}  Relative fitness of the stochastic gradient method in Algorithm~\ref{alg:1} for fraud detection after $T=100$ iterations versus the size of the datasets and the privacy budgets. }
\end{figure}

\subsubsection{Dataset Description} The dataset contains information regarding nearly 890,000 loans made on a peer-to-peer lending platform, called the Lending Club, which is available on Kaggle~\cite{kaggle1}. The inputs contain loan attributes, such as total loan size, and borrower information, such as number of credit lines, state of residence, and age. The outputs are the interest rates of the loans per annum. We encode categorical attributes, such as state of residence and loan grade assigned by the Loan Club, with integer numbers. We also remove unique identifier attributes, such as id and member id, as well as irrelevant attributes, such as the uniform resource locator (URL) for the Loan Club page with listing data. Finally, we perform feature selection using the Principal Component Analysis (PCA) to select the top ten important features. This step massively improves the numerical stability of the algorithm. For the PCA, we only use the last ten-thousand entries of the dataset to ensure that the feature selection does not violate the distributed nature of the algorithm. Note that, if we were to use the entire dataset for the PCA, the data should have been available at one location for processing which is contradictory to the assumptions of the paper regarding the distributed nature of the dataset and the privacy requirements of the data owners. After performing the PCA, the eigenvectors corresponding to the most important features are communicated to the distributed datasets. The first $n_1$ entries of the Lending Club are assumed to be the private data of the first data owner. The entries between $n_1+1$ to $n_1+n_2$ belong to the second data owner and the entries between $n_1+n_2+1$ to $n_1+n_2+n_3$ are with the third data owner. \farhad{Note that, by construct, these distributed datasets are non-overlapping, i.e., they do not share identical records.} We may use any other approach for splitting the Lending Club dataset among the private data owners as long as the distributed datasets are not overlapping. 
The data owners then balance their datasets using the-said eigenvectors. \farhad{The balancing refers \farhad{to a} transformation of the dataset using the eigenvectors to extract the most important independent features.} The eigenvectors, here, serve as a common dictionary between the data owners for communication and training.

\subsubsection{Experiment Setup} The experiments demonstrate the outcome of collaborations among $N=3$ financial institutes, e.g., banks, for training a ML model to automate the process of assigning interest rates to loan applications based on the attributes of the borrower and the loan. Each institute has access to a private dataset of $n_i$ historical loan applications and approved interest rates. The value of $\epsilon_i$ for each institute essentially determines eagerness for collaboration and openness to sharing private proprietary datasets. For a linear regression model, we consider a linear ML model relating the inputs and the outputs as in $y=\mathfrak{M}(x;\theta):=\theta^\top x$ with $\theta\in\mathbb{R}^{p_\theta}$ denoting the parameters of the ML model. We train the model by solving the optimization problem~\eqref{eqn:ML} with $g_2(\mathfrak{M}(x;\theta),y)= \|y-\mathfrak{M}(x;\theta)\|_2^2$, and $g_1(\theta)=0$.

\subsubsection{Results} First, we demonstrate the behaviour (e.g., convergence) of the iterates of the stochastic gradient descent procedure in Algorithm~\ref{alg:1}. Consider the case where $n_1=n_2=n_3=250,000$. Figure~\ref{fig:a0} shows the statistics of the relative fitness of the stochastic gradient method in Algorithm~\ref{alg:1} for a ML model determining lending interest rates, $\psi(\bar{\theta}[k])$, versus the iteration number $k$ for $T=100$ for three choices of privacy budgets $\epsilon_1=\epsilon_2=\epsilon_3$ \farhad{to illustrate the convergence of the learning algorithm as established in Theorem~\ref{tho:2}}. The algorithm is stochastic because the data owners provide differentially-private responses to the gradient queries, obfuscated with Laplace noise in Theorem~\ref{tho:1}. Thus each run of the algorithm follows a different relative fitness trend. The boxes, i.e., the vertical lines at each iterations, illustrate the range of 25\% to 75\% percentiles of the relative fitness extracted from one-hundred runs of the algorithm. The black lines show the median relative fitness versus the iteration number. The effect of the privacy budgets  on the quality of the iterates at the end of $T$ iterations is evident, as expected from Theorem~\ref{tho:2}. As $\epsilon_1=\epsilon_2=\epsilon_3$ increases, i.e., the data owners become more willing to share data, the performance of the trained ML model improves. \farhad{For instance, by increasing the privacy budget from $\epsilon_1=\epsilon_2=\epsilon_3=1$ to $\epsilon_1=\epsilon_2=\epsilon_3=10$, the relative fitness of the algorithm improves (i.e., decrease\farhad{s}), \farhad{on} average, by approximately 100-fold.}

After establishing the desired transient behaviour of the algorithm, we can investigate the effect of the size of the datasets and the privacy budgets on the performance of the trained ML model, i.e., the ML model after all the iterations have passed. Figure~\ref{fig:a1} shows the expectation (i.e., the statistical mean) of the relative fitness of the stochastic gradient method in Algorithm~\ref{alg:1} for the trained ML model after $T=100$ iterations versus the size of the datasets $n_1=n_2=n_3$ and the privacy budgets $\epsilon_1=\epsilon_2=\epsilon_3$. As predicted by Theorem~\ref{tho:2}, the fitness improves as the size of the datasets $n_1=n_2=n_3$ and/or the privacy budgets  $\epsilon_1=\epsilon_2=\epsilon_3$ increase. To quantify the tightness of the upper-bound in Theorem~\ref{tho:2} for Algorithm~\ref{alg:1}, we isolate the effects of the size of the datasets and the privacy budgets  on the relative fitness. Figure~\ref{fig:a3} illustrates the expectation of the relative fitness of the stochastic gradient method in Algorithm~\ref{alg:1} after $T=100$ iterations versus the privacy budgets  $\epsilon_1=\epsilon_2=\epsilon_3$. In this figure, the markers (i.e., {\small$\blacksquare$}, $\blacklozenge$, and \tikz\draw[black,fill=black] (0,0) circle (.6ex);) are from the experiments and the solid lines are fitted to the experimental data. 
We can see that the slope of the linear lines in the log-log scale in Figure~\ref{fig:a3} is $-2$. This shows that $\psi(\bar{\theta}[k])\propto \epsilon_i^{-2}$. Hence, our bound in Theorem~\ref{tho:2} is not tight as it states that  $\psi(\bar{\theta}[k])$ is upper bounded by a function of the form $\epsilon_i^{-1}$. This is because Theorem~\ref{tho:2} does not use the fact that the cost function for the regression is strongly convex and has Lipschitz gradients. These assumptions are utilized in Theorem~\ref{tho:utility_strongconvex} and the bounds in this theorem are in fact tight, as Theorem~\ref{tho:utility_strongconvex} states that  $\psi(\bar{\theta}[k])$ is upper bounded by a function of the form $\epsilon_i^{-2}$. Figure~\ref{fig:a4} shows the expectation of the relative fitness of the stochastic gradient method in Algorithm~\ref{alg:1} after $T=100$ iterations versus the size of the datasets $n_1=n_2=n_3$. Similarly, the slop of the linear lines in the log-log scale in Figure~\ref{fig:a4} is $-2$ pointing to that $\psi(\bar{\theta}[k])\propto n_i^{-2}$. This is again a perfect match for our theoretical bound in Theorem~\ref{tho:utility_strongconvex} (because $n=n_1+n_2+n_3=3n_i$). 

\begin{figure}[t]
\centering
\begin{tikzpicture}
\node[] at (0,0) {
\psfrag{x1}[cc][][0.6][0]{$10^{-1}$}
\psfrag{x2}[cc][][0.6][0]{$10^{0}$}
\psfrag{x3}[cc][][0.6][0]{$10^{1}$}
\psfrag{y1}[cc][][0.6][0]{$10^{-1}$\quad}
\psfrag{y2}[cc][][0.6][0]{$10^{1}$\quad}
\psfrag{y3}[cc][][0.6][0]{$10^{3}$\quad}
\psfrag{y4}[cc][][0.6][0]{$10^{5}$\quad}
\psfrag{y5}[cc][][0.6][0]{$10^{7}$\quad}
\psfrag{y6}[cc][][0.6][0]{$10^{9}$\quad}
\psfrag{aaaaaaaaaa1}[l][t][0.6][0]{\small \hspace{-.26in}\raisebox{-10pt}{$n=10^3$}}
\psfrag{aaaaaaaaaa2}[l][t][0.6][0]{\small \hspace{-.26in}\raisebox{-10pt}{$n=10^4$}}
\psfrag{aaaaaaaaaa3}[l][t][0.6][0]{\small \hspace{-.26in}\raisebox{-10pt}{$n=10^5$}}
\includegraphics[width=.75\columnwidth]{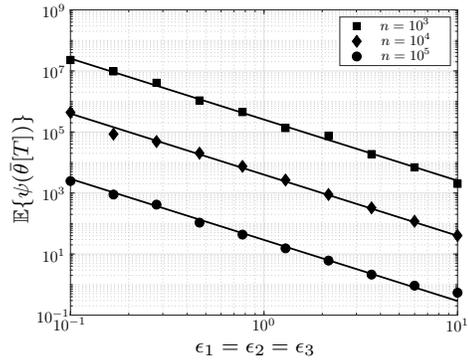}};
\node[] at (0,-2.4) {\footnotesize $\epsilon_1=\epsilon_2=\epsilon_3$};
\node[rotate=90] at (-3.1,0) {\footnotesize $\mathbb{E}\{\psi(\bar{\theta}[T])\}$};
\end{tikzpicture}
\vspace{-.1in}
\caption{\label{fig:b3} Relative fitness of the stochastic gradient method in Algorithm~\ref{alg:1} for fraud detection after $T=100$ iterations versus the privacy budgets. The solid line illustrate the bound in Theorem~\ref{tho:utility_strongconvex}. }
\end{figure}

\begin{figure}[t]
\centering
\begin{tikzpicture}
\node[] at (0,0) {
\psfrag{x1}[cc][][0.6][0]{$10^{2}$}
\psfrag{x2}[cc][][0.6][0]{$10^{3}$}
\psfrag{x3}[cc][][0.6][0]{$10^{4}$}
\psfrag{y1}[cc][][0.6][0]{$10^{-1}$\quad}
\psfrag{y2}[cc][][0.6][0]{$10^{1}$\quad}
\psfrag{y3}[cc][][0.6][0]{$10^{3}$\quad}
\psfrag{y4}[cc][][0.6][0]{$10^{5}$\quad}
\psfrag{y5}[cc][][0.6][0]{$10^{7}$\quad}
\psfrag{y6}[cc][][0.6][0]{$10^{9}$\quad}
\psfrag{aaaaaaa1}[l][t][0.6][0]{\small \hspace{-.24in}\raisebox{-10pt}{$\epsilon=0.1$}}
\psfrag{aaaaaaa2}[l][t][0.6][0]{\small \hspace{-.24in}\raisebox{-10pt}{$\epsilon=1$}}
\psfrag{aaaaaaa3}[l][t][0.6][0]{\small \hspace{-.24in}\raisebox{-10pt}{$\epsilon=10$}}
\includegraphics[width=.75\columnwidth]{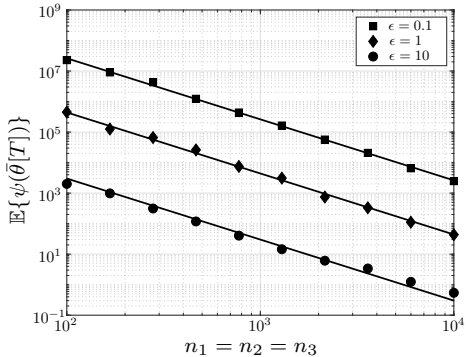}};
\node[] at (0,-2.4) {\footnotesize $n_1=n_2=n_3$};
\node[rotate=90] at (-3.1,0) {\footnotesize $\mathbb{E}\{\psi(\bar{\theta}[T])\}$};
\end{tikzpicture}
\vspace{-.1in}
\caption{\label{fig:b4} Relative fitness of the stochastic gradient method in Algorithm~\ref{alg:1} for fraud detection after $T=100$ iterations versus the size of the datasets. The solid line illustrate the bound in Theorem~\ref{tho:utility_strongconvex}.  }
\end{figure}

Finally, we consider a few scenarios of collaboration for the data owners. Specifically, we evaluate the performance of the learning algorithm for four distinct scenarios in which the second and the third data owners have: (\textit{i}) small datasets and small  privacy budgets (i.e., reluctant to share due to privacy concerns); (\textit{ii}) small datasets and large  privacy budgets  (i.e., eager to share); (\textit{iii}) large datasets and small  privacy budgets; (\textit{iv}) large datasets and large  privacy budgets. For each case, we vary the  privacy budget  and the size of the dataset of the first data owner. This allows us to investigate the potential benefit to data owners \farhad{from collaboration} in various scenarios. Figure~\ref{fig:a5} illustrates the expectation of the relative fitness of the stochastic gradient method in Algorithm~\ref{alg:1}, after $T=100$ iterations, versus the size of the dataset $n_1$ and the  privacy budget  $\epsilon_1$ for four distinct scenarios of collaboration. The first scenario in Figure~\ref{fig:a5} (the left most plot) shows that there is no point in collaboration with small data owners, even if the size of the dataset of the first data owner is large and it is eager to share its data\farhad{; the relative fitness (capturing the distance between private ML model and the non-private model) is very large, it does not change significantly with $\epsilon_1$, and it still remains large for relative large datasets $n_1=10^5$.} We could foresee this from the bound in Theorem~\ref{tho:2} without running Algorithm~\ref{alg:1}. This bound shows that $\psi(\bar{\theta}[k])\propto 1/(2000+n_1)\sqrt{200+1/\epsilon_1^2}$; hence, no matter how large $\epsilon_1$ gets (even if $\epsilon_1=\infty$), the error's coefficient remains large due to small  privacy budgets  of the other two data owners and $n_1$ must become considerably large to compensate for it. In the second scenario (the second left most plot in Figure~\ref{fig:a5}), the effect of $\epsilon_1$ and $n_1$ are more pronounced. This is because, although the other two data owners are small, they do not hinder the learning process by adding large amounts of privacy-preserving noise because of their conservatively small  privacy budgets. The third scenario is similar to the first one, albeit with better relative fitness as conservative data owners are relatively larger. The best scenario for collaboration, unsurprisingly, is the fourth scenario in which phenomenal performances can be achieved even without much consideration towards the size of the first dataset or its  privacy budget  as the other two datasets are large and eager to collaborate for learning.

\begin{figure*}[t]
\centering
\begin{tabular}{cccc}
\hspace{-.2in}
\footnotesize 
Scenario 1:
\hspace{-.2in}
&
\hspace{-.2in}
\footnotesize 
Scenario 2:
\hspace{-.2in}
&
\hspace{-.2in}
\footnotesize 
Scenario 3:
\hspace{-.2in}
&
\hspace{-.2in}
\footnotesize 
Scenario 4:\\[-.5em]
\hspace{-.2in}
\footnotesize 
$n_2=n_3=10^2$
\hspace{-.2in}
&
\hspace{-.2in}
\footnotesize 
$n_2=n_3=10^2$
\hspace{-.2in}
&
\hspace{-.2in}
\footnotesize 
$n_2=n_3=10^4$
\hspace{-.2in}
&
\hspace{-.2in}
\footnotesize 
$n_2=n_3=10^4$\\[-.4em]
\hspace{-.2in}
\footnotesize 
(small dataset)
\hspace{-.2in}
&
\hspace{-.2in}
\footnotesize 
(small dataset)
\hspace{-.2in}
&
\hspace{-.2in}
\footnotesize 
(large dataset)
\hspace{-.2in}
&
\hspace{-.2in}
\footnotesize 
(large dataset)\\[-.5em]
\hspace{-.2in}
\footnotesize 
$\epsilon_2=\epsilon_3=0.1$ 
\hspace{-.2in}
&
\hspace{-.2in}
\footnotesize 
$\epsilon_2=\epsilon_3=10$ 
\hspace{-.2in}
&
\hspace{-.2in}
\footnotesize 
$\epsilon_2=\epsilon_3=0.1$ 
\hspace{-.2in}
&
\hspace{-.2in}
\footnotesize 
$\epsilon_2=\epsilon_3=10$ \\[-.5em]
\hspace{-.2in}
\footnotesize 
(small privacy budget)
\hspace{-.2in}
&
\hspace{-.2in}
\footnotesize 
(large privacy budget)
\hspace{-.2in}
&
\hspace{-.2in}
\footnotesize 
(small privacy  budget)
\hspace{-.2in}
&
\hspace{-.2in}
\footnotesize 
(large privacy budget) \\[-.5em]
\hspace{-.2in}
\begin{tikzpicture}
\node[] at (0,0) {
\psfrag{x1}[cc][][0.6][0]{$10^3$}
\psfrag{x2}[cc][][0.6][0]{$10^4$}
\psfrag{x3}[cc][][0.6][0]{\quad$10^5$}
\psfrag{y1}[cc][bl][0.6][0]{\;$10^{-1}$}
\psfrag{y2}[cc][][0.6][0]{$10^{0}$}
\psfrag{y3}[cc][][0.6][0]{$10^{1}$\quad}
\psfrag{z1}[cc][t][0.6][0]{$10^{3}$\quad}
\psfrag{z2}[cc][][0.6][0]{$10^{4}$\quad}
\psfrag{z3}[cc][][0.6][0]{$10^{5}$\quad}
\psfrag{z4}[cc][][0.6][0]{$10^{6}$\quad}
\psfrag{z5}[cc][][0.6][0]{$10^{8}$\quad}
\includegraphics[width=.25\linewidth]{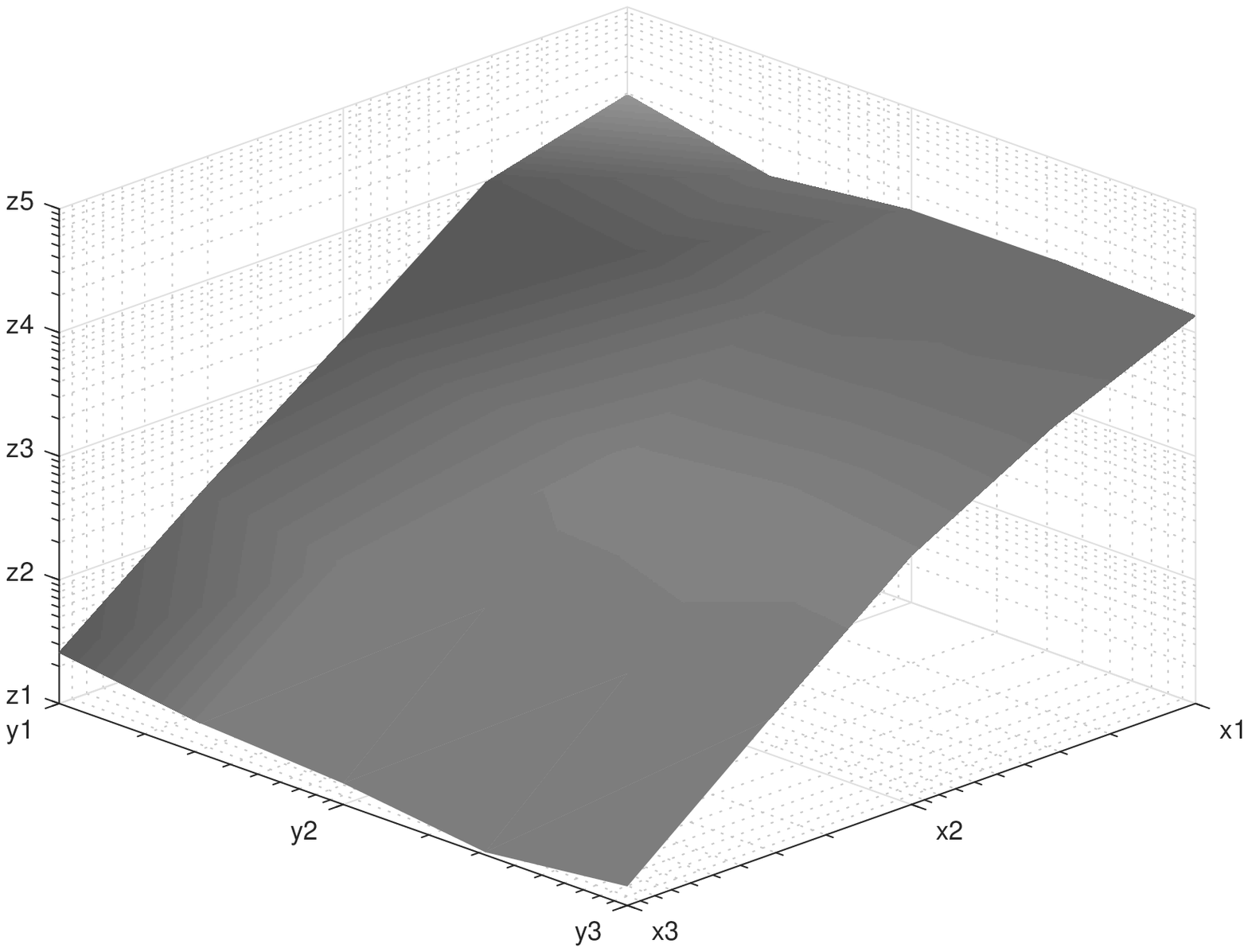}};
\node[] at (-1.1,-1.4) {\footnotesize $\epsilon_1$};
\node[] at (+1.2,-1.4) {\footnotesize $n_1$};
\node[rotate=90] at (-2.3,0) {\footnotesize $\mathbb{E}\{\psi(\bar{\theta}[T])\}$};
\end{tikzpicture}
\hspace{-.2in}
&
\hspace{-.2in}
\begin{tikzpicture}
\node[] at (0,0) {
\psfrag{x1}[cc][][0.6][0]{$10^3$}
\psfrag{x2}[cc][][0.6][0]{$10^4$}
\psfrag{x3}[cc][][0.6][0]{\quad$10^5$}
\psfrag{y1}[cc][bl][0.6][0]{\;$10^{-1}$}
\psfrag{y2}[cc][][0.6][0]{$10^{0}$}
\psfrag{y3}[cc][][0.6][0]{$10^{1}$\quad}
\psfrag{z1}[cc][t][0.6][0]{$10^{0}$\quad}
\psfrag{z2}[cc][][0.6][0]{$10^{2}$\quad}
\psfrag{z3}[cc][][0.6][0]{$10^{4}$\quad}
\psfrag{z4}[cc][][0.6][0]{$10^{6}$\quad}
\psfrag{z5}[cc][][0.6][0]{$10^{8}$\quad}
\includegraphics[width=.25\linewidth]{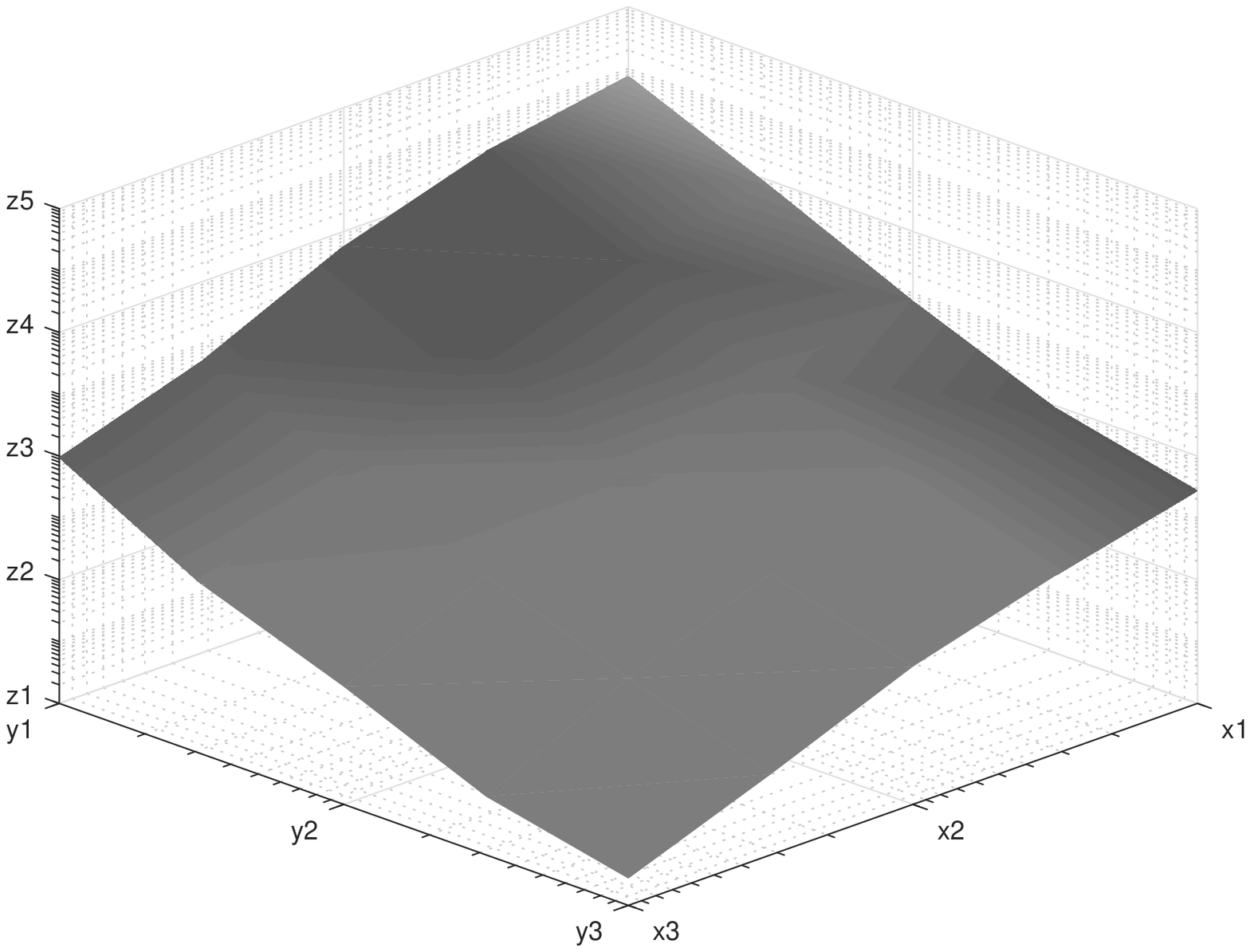}};
\node[] at (-1.1,-1.4) {\footnotesize $\epsilon_1$};
\node[] at (+1.2,-1.4) {\footnotesize $n_1$};
\node[rotate=90] at (-2.3,0) {\footnotesize $\mathbb{E}\{\psi(\bar{\theta}[T])\}$};
\end{tikzpicture}
\hspace{-.2in}
&
\hspace{-.2in}
\begin{tikzpicture}
\node[] at (0,0) {
\psfrag{x1}[cc][][0.6][0]{$10^3$}
\psfrag{x2}[cc][][0.6][0]{$10^4$}
\psfrag{x3}[cc][][0.6][0]{\quad$10^5$}
\psfrag{y1}[cc][bl][0.6][0]{\;$10^{-1}$}
\psfrag{y2}[cc][][0.6][0]{$10^{0}$}
\psfrag{y3}[cc][][0.6][0]{$10^{1}$\quad}
\psfrag{z1}[cc][t][0.6][0]{$10^{3}$\quad}
\psfrag{z2}[cc][][0.6][0]{$10^{4}$\quad}
\includegraphics[width=.25\linewidth]{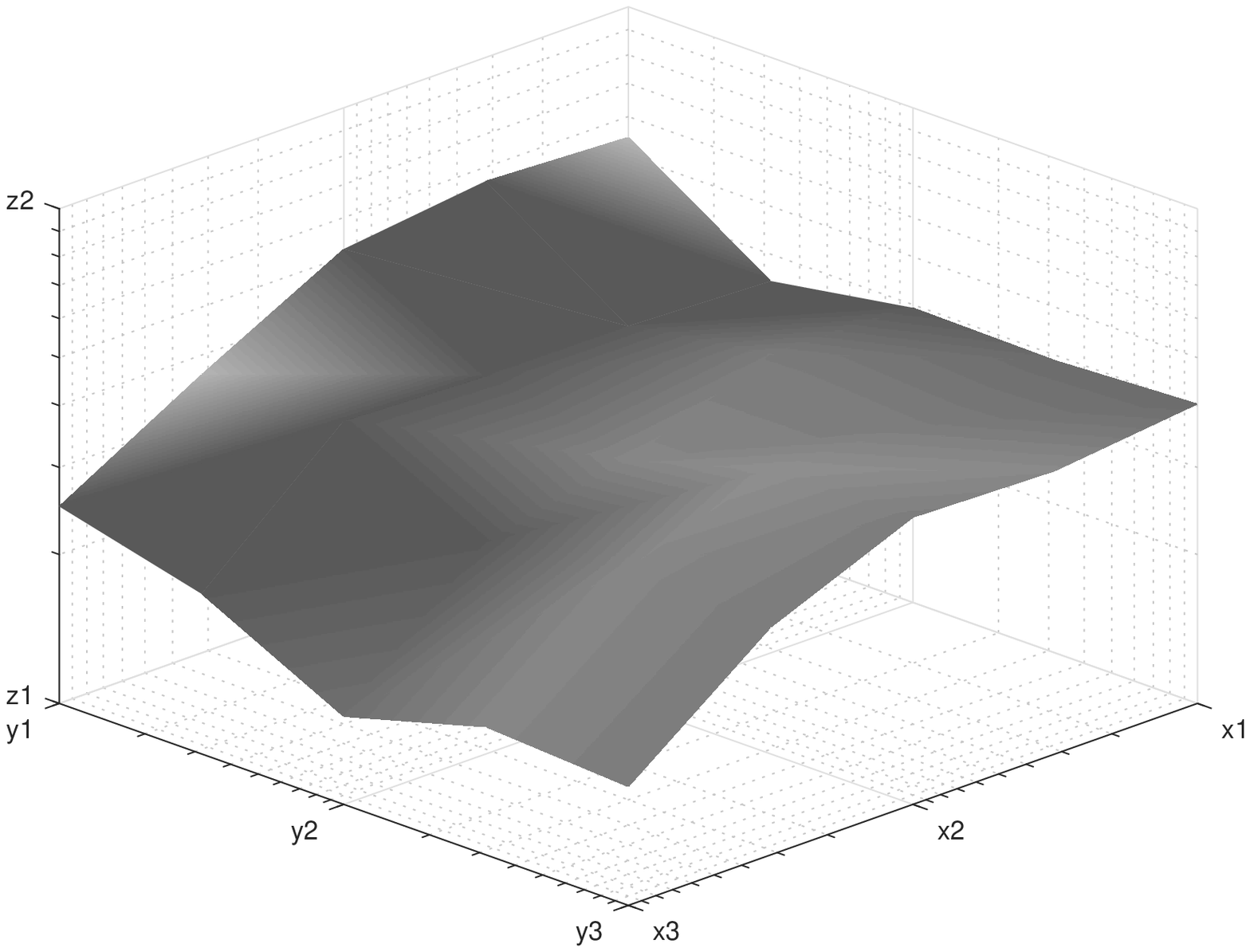}};
\node[] at (-1.1,-1.4) {\footnotesize $\epsilon_1$};
\node[] at (+1.2,-1.4) {\footnotesize $n_1$};
\node[rotate=90] at (-2.3,0) {\footnotesize $\mathbb{E}\{\psi(\bar{\theta}[T])\}$};
\end{tikzpicture}
\hspace{-.2in}
&
\hspace{-.2in}
\begin{tikzpicture}
\node[] at (0,0) {
\psfrag{x1}[cc][][0.6][0]{$10^3$}
\psfrag{x2}[cc][][0.6][0]{$10^4$}
\psfrag{x3}[cc][][0.6][0]{\quad$10^5$}
\psfrag{y1}[cc][bl][0.6][0]{\;$10^{-1}$}
\psfrag{y2}[cc][][0.6][0]{$10^{0}$}
\psfrag{y3}[cc][][0.6][0]{$10^{1}$\quad}
\psfrag{z1}[cc][t][0.6][0]{$10^{-1}$\quad}
\psfrag{z2}[cc][][0.6][0]{$10^{0}$\quad}
\psfrag{z3}[cc][][0.6][0]{$10^{1}$\quad}
\psfrag{z4}[cc][][0.6][0]{$10^{2}$\quad}
\psfrag{z5}[cc][][0.6][0]{$10^{3}$\quad}
\psfrag{z6}[cc][][0.6][0]{$10^{4}$\quad}
\includegraphics[width=.25\linewidth]{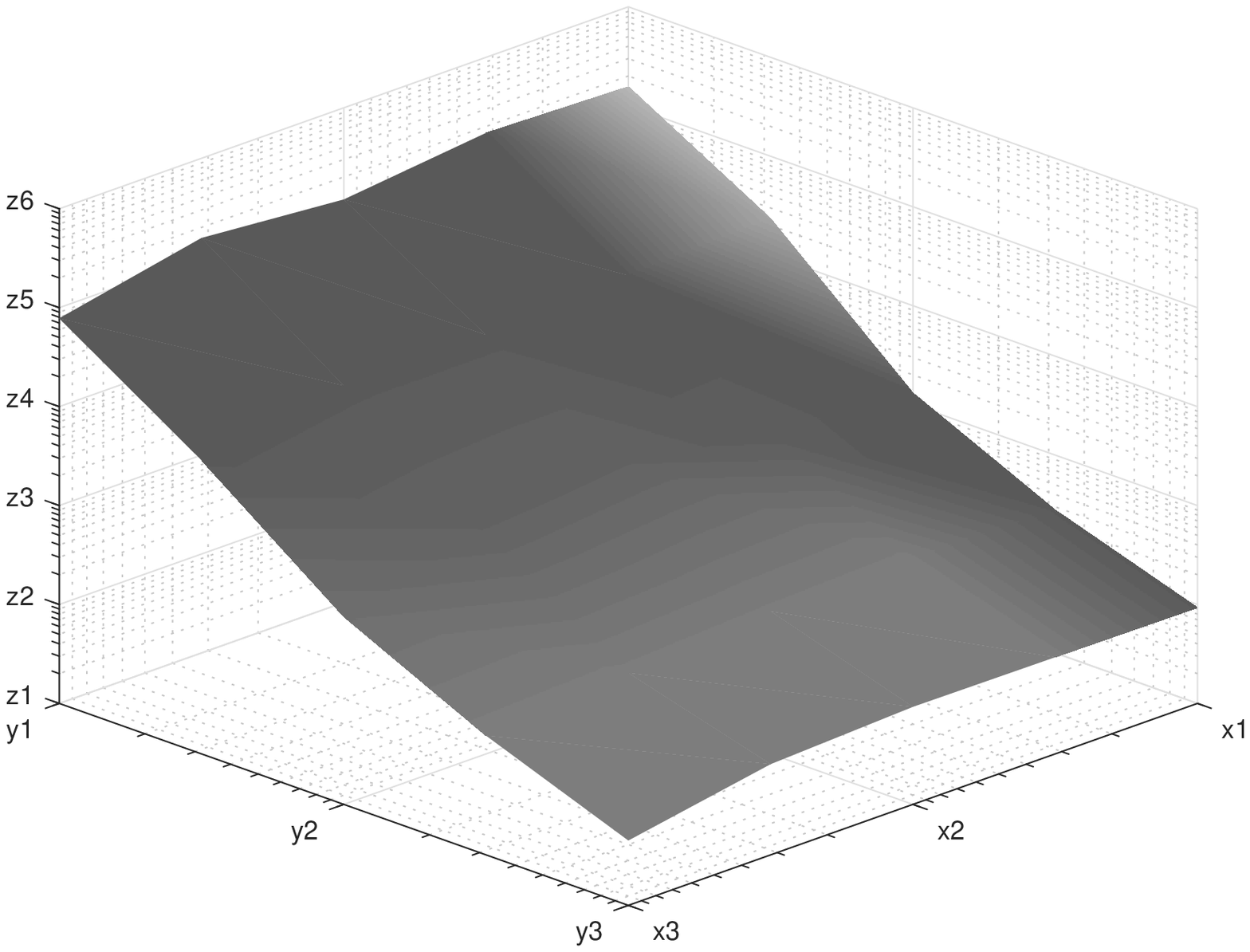}};
\node[] at (-1.1,-1.4) {\footnotesize $\epsilon_1$};
\node[] at (+1.2,-1.4) {\footnotesize $n_1$};
\node[rotate=90] at (-2.3,0) {\footnotesize $\mathbb{E}\{\psi(\bar{\theta}[T])\}$};
\end{tikzpicture}
\end{tabular}
\vspace{-.1in}
\caption{
\label{fig:b5}  
Relative fitness of the stochastic gradient method in Algorithm~\ref{alg:1} for a trained ML model determining lending interest rates after $T=100$ iterations versus the size of the dataset and the privacy budget of the first data owner for four distinct scenarios of collaboration.
 }
\end{figure*}

\subsection{Credit Card Fraud Detection}
In this subsection, we use a credit card dataset with a L-SVM classifier to further demonstrate the value of the methodology and to validate the theoretical results. 

\subsubsection{Dataset Description} The datasets contains transactions made by European credit card holders in September 2013 available on Kaggle~\cite{kaggle2}. The inputs are vectors extracted by PCA (to avoid confidentiality issues) as well as the amount of the transaction. The output is a class, determining if the transactions was deemed fraudulent or not. The dataset is highly unbalanced, as the positive class (frauds) account for 0.172\% of all transactions. 

\subsubsection{Experiment Setup} The experiments demonstrate the outcome of collaborations among $N=3$ financial institutes for training a SVM classifier to detect fraudulent activities automatically and rapidly. Each institute has access to a private dataset of $n_i$ historical credit card transactions and  their authenticity. The value of $\epsilon_i$ for each institute determines eagerness for collaboration. In L-SVM, the model is  $\mathfrak{M}(x;\theta):=\theta^\top [x^\top \; 1 ]^\top$, and $g_1(\theta):=(1/2)\theta^\top \theta $ and $ g_2(\mathfrak{M}(x;\theta),y):=\max(0,1-\mathfrak{M}(x;\theta)y).$

\subsubsection{Results} First, we investigate the transient behaviour of the iterates of  Algorithm~\ref{alg:1}. Assume that $n_1=n_2=n_3=30,000$. Figure~\ref{fig:b0} shows the statistics of the relative fitness of the iterates of Algorithm~\ref{alg:1} for training a fraud detection SVM classifier, $\psi(\bar{\theta}[k])$, versus the iteration number $k$ for $T=100$ for three choices of privacy  budgets $\epsilon_1=\epsilon_2=\epsilon_3$. The boxes, i.e., the vertical lines at each iterations, illustrate the range of 25\% to 75\% percentiles of relative fitness extracted from one-hundred runs of the algorithm and the black lines show the median relative fitness. As expected from Theorem~\ref{tho:2}, the performance of the trained SVM classifier gets closer to the SVM classifier trained with no privacy constraints $\theta^*$ as the privacy  budgets increases.

Now, we can demonstrate the effect of the size of the datasets and the privacy  budgets on the performance of the trained SVM classifier at the end of $T$ training iterations. Figure~\ref{fig:b1} shows the expectation of the relative fitness of the stochastic gradient method in Algorithm~\ref{alg:1} after $T=100$ iterations versus the size of the datasets $n_1=n_2=n_3$ and the privacy  budgets $\epsilon_1=\epsilon_2=\epsilon_3$. Similar to the theoretical results in Theorem~\ref{tho:2}, the fitness improves by increasing the size of the datasets $n_1=n_2=n_3$ and the privacy  budgets $\epsilon_1=\epsilon_2=\epsilon_3$. We can also isolate the effects of the size of the datasets and the privacy  budgets. Figure~\ref{fig:b3} illustrates the expectation of the relative fitness of the iterates of Algorithm~\ref{alg:1} after $T=100$ iterations versus the privacy  budgets $\epsilon_1=\epsilon_2=\epsilon_3$. As all linear slopes in the log-log scale in Figure~\ref{fig:b3} are $-2$, the bound in Theorem~\ref{tho:utility_strongconvex} seems to be a perfect fit. Figure~\ref{fig:a4} shows the expectation of the relative fitness of the iterates of Algorithm~\ref{alg:1} after $T=100$ iterations versus the size of the datasets $n_1=n_2=n_3$ revealing the exact behaviour predicted in the bound in Theorem~\ref{tho:utility_strongconvex}. 

Finally, we evaluate the performance of the learning algorithm for four distinct scenarios, in which the second and the third data owners have: (\textit{i}) small datasets and small privacy  budgets; (\textit{ii}) small datasets and large privacy  budgets; (\textit{iii}) large datasets and small privacy  budgets; (\textit{iv}) large datasets and large privacy  budgets. Figure~\ref{fig:b5} illustrates the expectation of the relative fitness of Algorithm~\ref{alg:1} after $T=100$ iterations versus the size of the dataset $n_1$ and the privacy  budget $\epsilon_1$ for four distinct scenarios of collaboration. The first scenario in Figure~\ref{fig:b5} (the left most plot) illustrates that there is no point in collaboration with small data owners even if the size of the dataset of the first data owner is large and it is eager to share its data.  In the second scenario (the second left most plot in Figure~\ref{fig:b5}), the effect of $\epsilon_1$ and $n_1$ are more pronounced because the privacy  budgets of the second and the third data owners are large and thus they do not degrade the performance of the learning algorithm by injecting excessive privacy-preserving noise.
The third scenario is again similar to the first one, albeit with better results as conservative data owners are relatively larger. The best scenario for collaboration, similar to the loan example, is the fourth scenario in which the training performances with and without privacy constraints are identical, so long as the dataset of the first subsystem is large, or its privacy  budget is not too small.

\section{Discussions, Conclusions, and Future Research}
\label{sec:conclusions}
We considered privacy-aware optimization-based ML on distributed private datasets. We assumed that the data owners provide DP responses to gradient queries. The theoretical analysis of the proposed DP gradient descent algorithms provided a way for predicting the quality of ML models based on the  privacy budgets and the size of the datasets. We proved that the difference between the training model with and without considering privacy constrains of the data owners is bounded by $(\sum_{\ell\in\mathcal{N}} n_\ell)^{-2}\sum_{\ell\in\mathcal{N}}\epsilon_\ell^{-2}$ in our  proposed algorithms under smoothness and strong-convexity assumptions for the fitness cost. The empirical results with real-world financial datasets split between multiple institutes/banks while using regression and support vector machine models demonstrated that the relative fitness in fact follows $\epsilon_i^{-2}$ and $n_i^{-2}$ for the proposed algorithm. This shows the tightness of the upper bounds on the difference between the trained ML models with and without privacy constraints from the theoretical analysis, which can be utilized for quantification of the privacy-utility trade-off in privacy-preserving ML. 

\farhad{Note that the data owners, themselves, can also play the role of the learner in Figure~\ref{fig:0}. In this case, the data owner who is interested in learning a model can query the other data owners to provide DP gradients to use for learning. Now, in this case, as the other data owners cannot access the trained model or the query responses, the data owner who is training the model can set its own privacy budget to infinity. Following this approach, by creating $N$ copies of the algorithm discussed in this paper, we can remove the central learner and each data owner can learn its own ML model. }

The results of this paper can be used or extended in multiple directions for future research:
\begin{itemize}
    \item We can extend the framework to multiple learners aiming to train separate privacy-aware ML models with similar structures based on their own datasets and DP responses from other learners and private data owners. This is closer in nature to the distributed or federated ML framework over an arbitrary connected communication network. Note that, in this paper, the communication structure among the learner and the data owners is over a star graph with the learner  at the center. 
    \item The results of this paper can be used to understand the behaviour of data owners and learners in a data market for ML training. The utility-privacy trade-off in this paper, in terms of the quality of the trained ML models, can be used in conjunction with the cost of sharing private data of costumers with the learner (in terms of loss of reputation, legal costs, implementation of privacy-preserving mechanisms, and communication infrastructure) to setup a game-theoretic framework for modeling interactions across a data market. The learner can compensate the data owners for access to their private data, by essentially paying them for choosing larger privacy budgets. After negotiations between the data owners and the learners for setting the privacy budgets, the algorithm of this paper can be used to then train ML models, while knowing in advance the expected quality of the trained model.
    \farhad{\item Synchronous updates of the algorithm is indeed a bottleneck of the proposed algorithm. Future work can focus on extending the results of this paper to asynchronous gradient updates where, at each iteration, only a subset of the data owners update the ML model. To be able to ensure the convergence of the asynchronous algorithm, we need to ensure that all the data owners update the model \farhad{as} frequently \farhad{as required}.}
    \farhad{\item Another direction for future research is to extend the framework of this paper to adversarial learning scenarios that can admit more general adversaries (\farhad{than the case of} curious-but-honest adversaries in this paper).}
\end{itemize}

\section*{\farhad{Acknowledgements}}
\farhad{We would like to thank Nicolas Papernot for shepherding our paper. His comments and suggestions greatly helped in improving the paper. The work has been funded, in part, by the ``Data Privacy in AI Platforms (DPAIP): Risks Quantification and Defence Apparatus'' project from the Next Generation Technologies Fund by the Defence Science and Technology (DST) in the Australian Department of Defence and the DataRing project funded by the NSW Cyber Security Network and Singtel Optus pty ltd through the Optus Macquarie University Cyber Security Hub.}

\bibliographystyle{ieeetr}
\bibliography{citation}

\appendices
\section{Proof of Theorem~\ref{tho:1}}
\label{proof:tho:1}
First, \farhad{because of~\eqref{eqn:diff_privacy_additive_noise}, we have}
\begin{align*}
\farhad{
    \|\overline{\mathfrak{Q}}_\ell(\mathcal{D}_\ell;k)-\overline{\mathfrak{Q}}_\ell(\mathcal{D}'_\ell;k))\|_1
=}&\farhad{\frac{1}{n_\ell}\Bigg\|\sum_{\{x,y\}\in\mathcal{D}_\ell}
\xi_{\bar{g}_2^{x,y}}(\theta[k])}\\
&\farhad{-\sum_{\{x,y\}\in\mathcal{D}'_\ell}
\xi_{\bar{g}_2^{x,y}}(\theta[k])\Bigg\|_1}\\
\farhad{=}&\farhad{\frac{1}{n_\ell}\| \xi_{\bar{g}_2^{x,y}}(\theta[k])|_{\{x,y\}\in\mathcal{D}_\ell\subseteq\mathcal{D}'_\ell}}\\
&\farhad{-\xi_{\bar{g}_2^{x,y}}(\theta[k])
|_{\{x,y\}\in\mathcal{D}'_\ell\subseteq\mathcal{D}_\ell}\|_1.}
\end{align*}
This implies that
$\|\overline{\mathfrak{Q}}_\ell(\mathcal{D}_\ell;k)-\overline{\mathfrak{Q}}_\ell(\mathcal{D}'_\ell;k))\|_1
\leq (2/n_\ell)\max_{\{x,y\}\in\mathcal{D}'_\ell\subseteq\mathcal{D}_\ell \cup \mathcal{D}_\ell\subseteq\mathcal{D}'_\ell}
\| \xi_{\bar{g}_2^{x,y}}(\theta[k])\|_1
\leq 2\Xi/n_\ell$\farhad{, where the last inequality follows from Assumption~\ref{assum:maximugradient}. Noting the exponential form of the Laplace random variable, we get} 
\begin{align*}
\farhad{\frac{p((\overline{\mathfrak{Q}}_\ell(\mathcal{D}_\ell;k))_{k=1}^T)}{p((\overline{\mathfrak{Q}_\ell}(\mathcal{D}'_\ell;k))_{k=1}^T)}}
&\farhad{=\hspace{-.03in}\prod_{k=1}^T\hspace{-.03in}\exp\hspace{-.03in}\Bigg(\hspace{-.03in}\frac{\|\overline{\mathfrak{Q}}_\ell(\mathcal{D}'_\ell;k)\|_1}{b}\hspace{-.03in}-\hspace{-.03in}\frac{\|\overline{\mathfrak{Q}}_\ell(\mathcal{D}_\ell;k)\|_1}{b} \hspace{-.05in}\Bigg)}\\
&\farhad{\leq \prod_{k=1}^T \exp(2\Xi/bn_\ell)}\\
&\farhad{=\exp(2\Xi T/bn_\ell),}
\end{align*}
where, by some abuse of notation, $p(\cdot)$ denotes the probability density of the variable in its argument. \farhad{Substituting $b=2\Xi T/(n_\ell \epsilon_\ell)$ in this inequality concludes the proof. }

\section{Proof of Theorem~\ref{tho:utility_strongconvex}}\label{proof:tho:utility_strongconvex}
\farhad{The magnitude of the DP noise is }
\begin{align*}
\mathbb{E}\{\|w[k]\|_2^2\}
=&
\mathbb{E}\bigg\{\bigg\| \bigg(\frac{1}{\sum_{\ell\in\mathcal{N}}n_j}\bigg)\sum_{j\in\mathcal{N}}n_\ell w_\ell[k]\bigg\|_2^2 \bigg\}\\
=&\bigg(\frac{1}{\sum_{\ell\in\mathcal{N}}n_j}\bigg)^2\sum_{\ell\in\mathcal{N}}n_\ell^2\mathbb{E}\{\|w_\ell[k]\|_2^2 \}\\
=&\bigg(\frac{1}{\sum_{\ell\in\mathcal{N}}n_j}\bigg)^2\sum_{\ell\in\mathcal{N}}\frac{8\Xi^2 T^2}{\epsilon_\ell^2}\\
=&\frac{8\Xi^2 T^2}{n^2}\sum_{\ell\in\mathcal{N}}\frac{1}{\epsilon_\ell^2}.
\end{align*}
Because $\nabla f$ is $\lambda$-Lipschitz, $f(z_1)\leq f(z_2)+\nabla f(z_2)^\top (z_1-z_2)+0.5\lambda\|z_2-z_1\|_2^2$ for all $z_1,z_2$~\cite{nesterov2013introductory} and therefore
\begin{align*}
\mathbb{E}\{f(\theta[k+1])\}\leq 
&\mathbb{E}\{f(\theta[k])\}\\
&+\mathbb{E}\{\nabla f(\theta[k])^\top(\theta[k+1]-\theta[k])\}\\
& +\frac{\lambda}{2}\mathbb{E}\{\|\theta[k+1]-\theta[k]\|_2^2\}\\
\leq &\mathbb{E}\{f(\theta[k])\}\\
&+\rho_k\bigg(\frac{\lambda\rho_k}{2}-1 \bigg)\mathbb{E}\{\|\nabla f(\theta[k])\|_2^2\}\\
& +\rho_k^2\frac{8\Xi^2 T^2}{n^2}\sum_{\ell\in\mathcal{N}}\frac{1}{\epsilon_\ell^2}.
\end{align*}
For all $\rho_k\leq 1/\lambda$, we have
\begin{align}
\mathbb{E}\{f(\theta[k+1])\}
\leq &\mathbb{E}\{f(\theta[k])\}-\frac{\rho_k}{2}\mathbb{E}\{\|\nabla f(\theta[k])\|_2^2\}\nonumber\\
& +\rho_k^2\frac{8\Xi^2 T^2}{n^2}\sum_{\ell\in\mathcal{N}}\frac{1}{\epsilon_\ell^2}.
\label{eqn:proof:newtho:1}
\end{align}
For $\varepsilon>0$, we may define
\begin{align*}
k_0:=\inf_{k} \bigg\{k\,\bigg|\, \mathbb{E}\{\|\nabla f(\theta[k])\|_2^2\}
\leq\frac{16\Xi^2 T^2\rho_k}{n^2}\sum_{\ell\in\mathcal{N}}\frac{1}{\epsilon_\ell^2}+\varepsilon \bigg\}.
\end{align*}
\farhad{Here, $k_0$ is the iteration number at which the magnitude of the last term in the right hand side of~\eqref{eqn:proof:newtho:1} (a positive value) becomes larger than the magnitude of the second to the last term in the right hand side of~\eqref{eqn:proof:newtho:1} (a negative value). In essence, at $k_0$, the upper bound on the cost function does not reduce.} If $T$ is large enough, we can easily show that there exists $k_0<\infty$. This can be proved by contrapositive. Assume that this not the case. Therefore,
\begin{align*}
\lim_{k\rightarrow 0}\mathbb{E}\{f(\theta[k])\}
=&\mathbb{E}\{f(\theta[1])\}\\
&+\sum_{t=2}^k(\mathbb{E}\{f(\theta[t])\}
-\mathbb{E}\{f(\theta[t-1])\})\\
\leq & \mathbb{E}\{f(\theta[1])\}-\sum_{t=2}^k\varepsilon \rho_k\\
=&-\infty.
\end{align*}
This is however not possible. Since $f$ is $L$-strongly convex, Polyak-Lojasiewicz inequality~\cite{nesterov2013introductory} implies that 
\begin{align*}
\mathbb{E}\{f(\theta[k_0])\}-f(\theta^*)\leq 
&\frac{1}{2L}\mathbb{E}\{\|\nabla f(\theta[k])\|_2^2\}\\
\leq & \frac{8\Xi^2 T^2\rho_k}{L n^2}\sum_{\ell\in\mathcal{N}}\frac{1}{\epsilon_\ell^2}+\frac{\varepsilon}{2L}.
\end{align*}
Now, because $k_0\leq T$, we get
\begin{align*}
\min_{1\leq k\leq T} \mathbb{E}\{f(\theta[k])\}-f(\theta^*)
\leq& \mathbb{E}\{f(\theta[k_0])\}-f(\theta^*)\\
\leq & \frac{8\Xi^2 T^2\rho_k}{L n^2}\sum_{\ell\in\mathcal{N}}\frac{1}{\epsilon_\ell^2}+\frac{\varepsilon}{2L}.
\end{align*}
Again, because $f$ is $L$-strongly convex, we can see that
\begin{align*}
f(\theta^*)
&\leq f(t\theta+(1-t)\theta^*) \\
&\leq tf(\theta)+(1-t)f(\theta^*)-\frac{L}{2}t(t-1)\big\|
\theta-\theta^*
 \big\|_2^2,
\end{align*}
for all $t\in(0,1)$. Setting $t=1/2$ results in
\begin{align} \label{eqn:proof:stronglyconvexinequality}
\big\|
\theta-\theta^*
 \big\|_2^2
\leq 4(f(\theta)-f(\theta^*))/L.
\end{align}
Hence,
\begin{align*}
\min_{1\leq k\leq T}\|\theta[k]-\theta^*\|_2^2\leq & \frac{4}{L}\bigg(\min_{1\leq k\leq T} \mathbb{E}\{f(\theta[k])\}-f(\theta^*)\bigg)\\
\leq & \frac{32\Xi^2 T^2\rho_k}{L^2 n^2}\sum_{\ell\in\mathcal{N}}\frac{1}{\epsilon_\ell^2}+\frac{\varepsilon}{8L^2}.
\end{align*}
This concludes the proof.

\section{Proof of Theorem~\ref{tho:2}}
\label{proof:tho:2}
\farhad{The proof for this theorem follows from modification of the results of~\cite{shamir2013stochastic}. In fact}, the inequality in~\eqref{eqn:1} follows from the result of~\cite{shamir2013stochastic} using the optimal selection of $c$ in~\cite{han2017differentially}. The only difference with the proofs in~\cite{shamir2013stochastic} is to appreciate that
 \begin{align*}
 \zeta_k-\zeta_{k-1}\leq \frac{2}{\sqrt{T}T(T+1)},
 \end{align*}
 where
 \begin{align*}
 \zeta_k:=
 &\frac{1/\sqrt{T}+1}{1/\sqrt{T}+k} \prod_{m=k+1}^T \frac{m-1}{1/\sqrt{T}+m}. 
 \end{align*}
The inequality follows from that
\begin{align*}
\zeta_k-\zeta_{k-1}
&=
\frac{(1/\sqrt{T})(1/\sqrt{T}+1)}{(k-1+1/\sqrt{T})(k+1/\sqrt{T})}
\\
&\hspace{.2in}\times\prod_{m=k+1}^T \frac{m-1}{1/\sqrt{T}+m} \\
&=
\frac{(1/\sqrt{T})(1/\sqrt{T}+1)}{(k-1+1/\sqrt{T})(k+1/\sqrt{T})}
\\
&\hspace{.2in}\times \frac{\prod_{m=k+1}^T (m-1)}{\prod_{m=k+1}^T(1/\sqrt{T}+m)} \\
 &=
(1/\sqrt{T})(1/\sqrt{T}+1)
 \frac{\prod_{m=k}^{T-1} m}{\prod_{m=k-1}^T(1/\sqrt{T}+m)} \\
  &=
(1/\sqrt{T})(1/\sqrt{T}+1)
 \frac{\prod_{m=k}^{T-1} m}{\prod_{m=k}^{T+1}(1/\sqrt{T}+m-1)} \\
   &=
\frac{(1/\sqrt{T})(1/\sqrt{T}+1)}{T(T+1)}
\prod_{m=k}^{T+1} \frac{ m}{(1/\sqrt{T}+m-1)}\\
& \leq \frac{2}{\sqrt{T}}\frac{1}{T(T+1)}.
\end{align*}
If $f$ is $L$-strongly convex, the proof of the inequality in~\eqref{eqn:2} follows from~\eqref{eqn:proof:stronglyconvexinequality}.

\end{document}